\documentclass[onefignum,onetabnum]{siamart171218}

\usepackage{bm}
\usepackage{lipsum}
\usepackage{amsfonts}
\usepackage{graphicx}
\usepackage{epstopdf}
\usepackage{algorithmic}
\usepackage{color}
\ifpdf
  \DeclareGraphicsExtensions{.eps,.pdf,.png,.jpg}
\else
  \DeclareGraphicsExtensions{.eps}
\fi

\graphicspath{{figures/}}

\newcommand{\drt}[1]{{\color{blue}#1}}

\newsiamremark{remark}{Remark}
\newsiamremark{hypothesis}{Hypothesis}
\crefname{hypothesis}{Hypothesis}{Hypotheses}
\newsiamthm{claim}{Claim}

\headers{Network Synchronization with Uncertain Dynamics}{P. S. Skardal, D. Taylor and J. Sun}

\title{{Synchronization of Network-Coupled Oscillators with Uncertain Dynamics}
\thanks{PSS and DT contributed equally to this work. DT was supported by the Simons Foundation, Award \#578333.}
}

\author{Per Sebastian Skardal\thanks{Department of Mathematics, Trinity College, Hartford, CT 06106 (\email{persebastian.skardal@trincoll.edu}, \url{https://sites.google.com/site/persebastianskardal/}).}
\and Dane Taylor\thanks{Department of Mathematics, University at Buffalo, State University of New York, Buffalo, New York 14260 (\email{danet@buffalo.edu}, \url{https://sites.google.com/site/danetaylorresearch/}).}
\and Jie Sun\thanks{Department of Mathematics, Clarkson University, Potsdam, New York 13699 (\email{sunj@clarkson.edu}, \url{http://people.clarkson.edu/\~sunj/}).}
}

\usepackage{amsopn}

\ifpdf
\hypersetup{
  pdftitle={\drt{Synchronization of Network-Coupled Oscillators with Uncertain Frequencies}},
  pdfauthor={P. S. Skardal, D. Taylor and J. Sun}
}
\fi

\begin{document}

\maketitle

\begin{abstract}
Synchronization of network-coupled dynamical units is important to a variety of natural and engineered processes including circadian rhythms, cardiac function,  neural processing, and power grids. Despite this ubiquity, it remains poorly understood how complex network structures and heterogeneous local dynamics combine to either promote or inhibit synchronization. Moreover, for most real-world applications it is  impossible to obtain the exact specifications of the system, and  there is a lack of theory for how uncertainty affects synchronization. We address this  open problem by studying the Synchrony Alignment Function (SAF), which is an objective measure for the synchronization properties of a network of heterogeneous oscillators with given natural frequencies. We extend the SAF framework to analyze network-coupled oscillators with heterogeneous natural frequencies that are drawn as a multivariate random vector. Using probability theory for quadratic forms, we obtain expressions for the expectation and variance of the SAF for given network structures. We conclude with numerical experiments  that illustrate how the incorporation of uncertainty yields a more robust theoretical framework for enhancing synchronization, and we provide new perspectives for why synchronization  is generically promoted by network properties including degree-frequency correlations, link directedness, and link weight delocalization.
\end{abstract}



\section{Introduction}\label{sec:01}

The emergence of collective behavior in ensembles of network-coupled dynamical systems is a widely studied and active area of research in the dynamical systems and network-science communities~\cite{Arenas2008PR,Pikovsky2003}. Examples where synchronization is vital for enabling robust functionality are plentiful in both natural and engineered systems. Examples of natural systems that display synchronization include cardiac pacemakers~\cite{Mirollo1990SIAP}, circadian rhythms~\cite{Winfree1967JTB}, neuronal processing~\cite{Medvedev2001SIAP}, gene regulation~\cite{Kuznetsov2004SIAP}, intestinal activity~\cite{Aliev2000JTB}, and cell cycles~\cite{Prindle2012Nature}. Examples of engineered systems that require robust synchronization include Josephson junction arrays~\cite{Wiesenfeld1996PRL}, electrochemical oscillators~\cite{Kazanci2007SIAP}, and power grids~\cite{Rohden2012PRL,Skardal2015SciAdv}. On the other hand, in certain systems the excess of synchronization may lead to unfavorable consequences, including bridge oscillations \cite{Strogatz2005Nature} and Parkinson's disease, tremors, and seizures~\cite{Schnitzler2005Nature}.

Given the widespread importance of robust synchronization in so many applications, a broad but practical question that still remains poorly understood is the following: How do the dynamical and structural properties in a network of coupled oscillators affect the macroscopic synchronization properties of the system as a whole? A large body of literature indicates that the microscopic properties of a network's topology and the local dynamics and coupling mechanism in oscillator networks combine in very complicated ways and can lead to a wide range of novel dynamical behaviors~ \cite{Gomez2011PRL,Restrepo2014EPL,Restrepo2005PRE,Skardal2012PRE,Skardal2015PRE3}. It is therefore natural to ask: What arrangements of network topologies and local oscillator dynamics {\it optimize} network synchronization? 

In previous work, we addressed this question by developing a framework that introduced the {\it Synchrony Alignment Function} (SAF) \cite{Skardal2014PRL,Skardal2016Chaos}, which is an objective function that quantifies the interplay between complex network structure and local dynamics. The SAF quantifies the alignment of the oscillators' heterogeneities with the network heterogeneity, as manifest in the singular vectors (or eigenvectors in the case of undirected graphs) of the combinatorial graph Laplacian matrix. In section~\ref{subsec:02:02}, we present the generalized version of the SAF \cite{Skardal2016Chaos}, which allows the network to possibly contain directed links. Importantly, the SAF gives rise to a nonnegative scalar that serves as an objective measure of the synchronization properties of an oscillator network and provides an avenue for optimization; it allows for the systematic optimization of synchronization properties for networks of heterogeneous oscillators under a wide range of practical constraints \cite{Skardal2014PRL,Skardal2016Chaos,Taylor2016SIAP,Skardal2016Chaos}. For example, the SAF framework has been used to rank network links according to their importance to synchronization (which for heterogeneous oscillators, depends both on the network and the oscillator frequencies), thereby identifying network modification that judiciously enhance synchronization~\cite{Taylor2016SIAP}. It has been applied to optimize heterogeneous chaotic oscillators and was shown to effectively improve the synchronization of experimental electronic circuits~\cite{Skardal2017Chaos}. The SAF framework has also served useful beyond the goal of optimization---namely, it has been used  to characterized a phenomenon called erosion of synchronization, which arises under  frustrated coupling~\cite{Skardal2015PRE2,Skardal2015PhysD}.

The SAF quantifies the alignment of network heterogeneity with the oscillators' heterogeneity---that is, their different natural frequencies $\{\omega_i\}$---which  requires knowledge of the precise network and the particular frequencies. However, such information can be impossible to obtain for many applications. For example, any attempt to measure the oscillation frequency for a brain region, heart tissue, intestinal tissue, etc. will inevitably produce some measurement error.  Understanding the effects of uncertainty on synchronization, and its optimization, remains an important open problem. Thus motivated, in this paper we relax the assumption that the frequencies are deterministic parameters. Instead, we extend the SAF framework by allowing the oscillators' natural frequencies to be encoded by a multivariate random variable $\bm{\omega}$ with known expectation $\bm{\mu}$ and covariance matrix $\Sigma$.  Using the observation that the SAF is a quadratic form of a nonnegative, symmetric matrix \cite{Rencher2008}, we utilize the extensive literature on  quadratic forms for random vectors  to obtain expressions for the expectation and variance of the SAF. Our results provide theoretical justification for numerical experiments presented in \cite{Taylor2016SIAP} (see figure 6.2(c)), where we found the SAF to effectively optimize systems with moderate uncertainty about the frequencies. We also report new and unexpected findings, for instance that for independent and identically distributed frequencies, uncertainty always increases the SAF, in expectation, thereby inhibiting synchronization.

In the simplest case, each  frequency $\omega_i$ is an identically and independently distributed random variable, implying $\bm{\mu}=\overline{\mu} [1,\dots,1]^T$ and $\Sigma = \sigma^2 I$ where $\overline{\mu}\in\mathbb{R}$ and $\sigma^2\in\mathbb{R}_+$ are constants. This scenario describes when one has information about the network structure and the distribution of frequencies, but no knowledge about how the oscillators' frequencies are arranged across the  network. We use this setting to shed light on another important topic in the field of synchronization: What generic structural and/or dynamical properties of oscillator networks promote or inhibit synchronization? We begin by demonstrating that synchronization is promoted when link weights are delocalized, i.e., when a network has more links that are on average weaker compared to fewer links that are on average stronger. We support this finding through both numerical experiments and analytically, by combining our probabilistic analysis of the SAF with random matrix theory that predicts the spectral properties of the combinatorial graph Laplacian matrix $L$. We find that synchronization is enhanced by a concentration (i.e., lack of heterogeneity) of the spectral density for $L$ (singular values in the case of directed networks and eigenvalues in the case of undirected networks), which establishes a new connection between synchronization theory for heterogeneous oscillators and identical oscillators \cite{Barahona2002PRL,Nishikawa2010PNAS,Pecora1998PRL,Sun2009EPL}. In particular, the eigenratio ---  which is one measure for spectral heterogeneity --- is well known to be important for identical oscillators \cite{Barahona2002PRL}. Here, we find for IID frequencies that the expectation and variance of the SAF  are proportional to the second and fourth moments of the spectral density, respectively, which are two different measures for spectral heterogeneity. We also show that increasing the directedness of a network (that is, the fraction of links that are directed versus undirected) promotes synchronization in a similar way as delocalization. Finally, we consider correlations between the node degrees and the oscillator frequencies, and complement previous research on the effects of these correlations \cite{Brede2008PLA,Papadopoulos2017Chaos,Pinto2015PRE,Skardal2014PRL,Skardal2016Chaos} by exploring how different types of correlations can be used to improve or diminish synchronization.

The remainder of this paper is organized as follows. In \cref{sec:02}, we present background information. In \cref{sec:03}, we present our main analytical results, quantifying the expectation and variance of the SAF using a latent variable approach for natural frequency ensembles. We also demonstrate the accuracy of these results using numerical simulations. In \cref{sec:04}, we apply these new results to identify generic network properties that promote synchronization, including link weight delocalization, directedness, and degree-frequency correlations. In \cref{sec:05}, we conclude with a discussion.

\section{Background Information and a Motivating Numerical Experiment}\label{sec:02}
We first present Kuramoto's  phase oscillator model (\cref{subsec:02:01}), the SAF framework (\cref{subsec:02:02}), and a motivating  experiment (\cref{subsec:02:03}).

\subsection{The Network Kuramoto Model}\label{subsec:02:01}

\begin{definition}[Network Kuramoto Phase Oscillator Model~\cite{Kuramoto}]\label{def:Kuramoto}
Consider a collection of $N$ phase oscillators in which $\theta_n\in[0,2\pi)$ is the phase of oscillator $n$, $\omega_n\in\mathbb{R}$ is the natural frequency of oscillator $n$, the adjacency matrix $A$ encodes the network-coupling of oscillators so that $A_{nm}$ is the weight of the link $m\to n$, $H_{nm} :\mathbb{R}\to\mathbb{R}$ is an interaction-specific, $2\pi$-periodic coupling function that is differentiable at $0$, and $K>0$ is the global coupling strength. The network Kuramoto phase oscillator model \cite{Kuramoto} is given by the following system of first-order nonlinear differential equations:
\begin{align}
\frac{d{\theta}_n}{dt}= \omega_n + K\sum_{m=1}^N A_{nm}H_{nm}(\theta_m-\theta_n),~~~\text{for }n\in\{1,\dots,N\}.\label{eq:02:01}
\end{align}
\end{definition}

Kuramoto's phase oscillator model was originally derived to describe the synchronization of {\it weakly-coupled} limit-cycle oscillators. This refers to oscillators whose steady-state behavior in the absence of coupling are stable limit cycles, and when coupling is added the geometry of the limit cycles are not significantly altered. Moreover, it is often assumed that the oscillator interactions follow a uniform functional form, i.e., $H_{nm}(\theta)=H(\theta)$ for all $n,m=1,\dots,N$. The choice $H(\theta)=\sin(\theta)$ is particularly popular and well representative of the generic case, using the first-order term of a Fourier expansion for an odd function $H(\theta)$, and is often referred to simply as the network Kuramoto model. Of particular importance is the network structure encoded by adjacency matrix $A$. The entries of $A$ indicate the existence and weight of each link in the network, such that $A_{nm}\ge0$ is the weight of the link $m\to n$. (If no link exists, then $A_{nm}=0$.) We assume that the network is weighted, has non-negative links, and no self loops. In the case of a directed network, we define the (weighted) in- and out-degrees by $k_n^{\text{in}}=\sum_{m=1}^NA_{nm}$ and $k_n^{\text{out}}=\sum_{m=1}^NA_{mn}$, respectively. For undirected networks, one has $A_{nm}=A_{mn}$, implying $A^T=A$, and we define $k_n^{\text{in}}=k_n^{\text{out}}=k_n$. We also define the mean degree $\langle k\rangle = (1/N)\sum_{n=1}^Nk_n$.

Next we define the classical Kuramoto order parameter, which is an objective measure of the degree of synchronization in the network Kuramoto model.

\begin{definition}[The Kuramoto Order Parameter~\cite{Kuramoto}]\label{def:OrderParameter}
Consider the network Kuramoto phase oscillator model given in \cref{def:Kuramoto}. The Kuramoto order parameter is the complex number $z\in\mathbb{C}$ given by the centroid of all $N$ oscillator phases $\{\theta_j\}$ when mapped onto the complex unit circle:
\begin{align}
z = \frac{1}{N}\sum_{j=1}^Ne^{i\theta_j},\label{eq:02:02}
\end{align}
where $i=\sqrt{-1}$ is the imaginary unit.
\end{definition}

It is convenient to use the polar decomposition of the order parameter $z=re^{i\psi}$, where $r=|z|\in[0,1]$ is the magnitude of $z$ and $\psi$ is the collective phase of the system. In particular, $r$ is an objective measure of the degree of synchronization since an incoherent state yields $r\approx0$, and a strongly synchronized state yields $r\approx1$.

\subsection{The Synchrony Alignment Function}\label{subsec:02:02}

Next, we summarize the 
SAF framework, which was first introduced in \cite{Skardal2014PRL} for undirected networks and later generalized for directed networks~\cite{Skardal2016Chaos}. To help facilitate the presentation of our main analytical results in section~\ref{sec:03}, we will define the SAF in terms of a quadratic form~\cite{Rencher2008}. The fact that the SAF takes this form has not been noted in previous research.

We begin with the assumption that the coupling function is uniform, $H_{nm}(\theta)=H(\theta)$, and that $H(\theta)$ is smooth and positively-sloped at zero, and $H(0)\ll1$ with a root $H(\theta)=0$ nearby. As shown/discussed in \cite{Skardal2015PRE2}, these assumptions allow for the possibility of a strongly synchronized state when the coupling strength $K$ is for a sufficiently large or when the spread of oscillator frequencies is sufficiently small. We then consider the strongly synchronized regime where phase angles satisfy $|\theta_m-\theta_n|\ll1$, allowing us to linearize \cref{eq:02:01} to obtain the linearized system
\begin{align}
\dot{\theta}_n=\tilde{\omega}_i -KH'(0)\sum_{j=1}^N L_{nm}\theta_m,\label{eq:02:03}
\end{align}
where $\tilde{\omega}_n=\omega_n+KH(0)k_n^{in}$ is the effective natural frequency of oscillator $n$ and $L$ is the combinatorial Laplacian matrix whose entries are given by $L_{nm}=\delta_{nm}k_n^{in}-A_{nm}$. 

The combinatorial graph Laplacian $L$ is a cornerstone for numerous problems in network science, and we will make use of its singular value decomposition. The singular value decomposition of $L$ is given by $L=US V^T$, where $S=\text{diag}(s_1,s_2,\dots,s_N)$, and the columns of $U$ and $V$ are populated by the orthonormal left and right singular vectors $\{\bm{u}^i\}_{i=1}^N$ and $\{\bm{v}^i\}_{i=1}^N$, which satisfy   $L\bm{v}^i=s_i\bm{u}^i$ and $L^T\bm{u}^i=s_i\bm{v}^i$~\cite{Golub}. 
For undirected networks $L$ is symmetric, $L=L^T$, and the singular values and singular vectors are identical to the eigenvalues and eigenvectors of $L$, and so it suffices to consider only the eigenvalue decomposition is required $L = V\Lambda V^T$, where $\Lambda=\text{diag}(\lambda_1,\lambda_2,\dots,\lambda_N)$ and the columns of $V$ are populated by the orthonormal eigenvectors.

The SAF framework utilizes the pseudo-inverse of $L$, which is generally defined by $L^\dagger = VS^\dagger U^T$, where $S^\dagger=\text{diag}(0,s_2^{-1},\dots,s_N^{-1})$~\cite{BenIsrael2003}. For undirected networks one has $L^\dagger = V\Lambda^\dagger V^T$ with $\Lambda^\dagger=\text{diag}(0,\lambda_2^{-1},\dots,\lambda_N^{-1})$.  Note that we explicitly assume there are $N-1$ singular values and nonzero eigenvalues, which occurs, for example, if the network is strongly connected.
Under this assumption, the real part of all other eigenvalues is positive, $\text{Re}(\lambda_i)>0$ for $i=2,\dots,N$ and all other singular values are positive and can be ordered $0=s_1<s_2\le\cdots\le s_N$.

Returning to the dynamics of \cref{eq:02:03},  we search for a phase-locked solution of the form $\theta_n(t)=\theta_n^*+\Omega t$, where $\Omega$ is the collective frequency of the synchronized state. In the case of an undirected network we have $\Omega=\langle\tilde{\bm{\omega}}\rangle$, whereas in the directed case we have $\Omega=(\bm{u}^{1T} \tilde{\omega})/(\bm{u}^{1T}\bm{1})$~\cite{Skardal2015PRE}. In either case, we enter the rotating reference frame $\theta_n\mapsto\theta_n+\Omega t$, and search for the phase-locked state in this rotating frame satisfying $\dot{\theta}_n=0$ for all $n=1,\dots,N$. The solution is given, in vector form, by
\begin{align}
\bm{\theta}^*=\frac{L^\dagger\tilde{\bm{\omega}}}{KH'(0)}.\label{eq:02:04}
\end{align}
Returning to the Kuramoto order parameter in \cref{eq:02:02}, in the strongly synchronized regime we have that
\begin{align}
r\approx1-\frac{\|\bm{\theta}^*\|^2}{2N}=1-\frac{\tilde{\bm{\omega}}^TL^{\dagger T}L^\dagger\tilde{\bm{\omega}}}{2NK^2H'^2(0)}=1-\frac{\tilde{\bm{\omega}}^TU(S^{\dagger})^2U^T\tilde{\bm{\omega}}}{2NK^2H'^2(0)},\label{eq:02:05}
\end{align}
where we have used that $V^TV=I$ and $S^\dagger$ is diagonal so that $(S^\dagger)^TS^\dagger=(S^\dagger)^2=\text{diag}(0,s_2^{-2},\dots,s_N^{-2})$. Equation~\cref{eq:02:05} brings us to the quadratic form definition of the SAF

\begin{definition}[Synchrony Alignment Function (SAF)~\cite{Skardal2014PRL,Skardal2016Chaos}]\label{def:02:03}
Consider the network Kuramoto phase oscillator model given in \cref{def:Kuramoto} with linearized dynamics given by \cref{eq:02:03} with effective natural frequencies $\tilde{\bm{\omega}}$ and Laplacian $L$. The Synchrony Alignment Function is a function $J:\mathbb{R}^N\times\mathbb{M}_{N\times N}\to\mathbb{R}$ defined as
\begin{align}
J(\tilde{\bm{\omega}},L)=\frac{\tilde{\bm{\omega}}^TU(S^{\dagger})^2U^T\tilde{\bm{\omega}}}{N}\label{eq:02:06}
\end{align}
where $(S^\dagger)^2=\text{diag}(0,s_2^{-2},\dots,s_N^{-2})$ is populated by the inverse squares of the nontrivial singular values of $L$ and the columns of $U$ are populated by the left singular vectors $\{\bm{u}^i\}_{i=1}^N$ of $L$.
\end{definition}

For the remainder of this paper, we will refer to \eqref{eq:02:06} simply as the SAF. We emphasize that the SAF is given here as a quadratic form with matrix $U(S^{\dagger})^2U^T$. It appears different from, although is equivalent to, the forms given in previous work~\cite{Skardal2014PRL,Skardal2016Chaos}. In particular, one can alternately write the pseudo-inverse of the Laplacian as $L^\dagger=\sum_{j=2}^N\bm{v}^j\bm{u}^{jT}/s_j$, which leads to the alternate form
\begin{align}
J(\tilde{\bm{\omega}},L)=\frac{1}{N}\sum_{j=2}^N\frac{(\bm{u}^{jT}\tilde{\bm{\omega}})^2}{s_j^2},\label{eq:02:07}
\end{align}
which more explicitly describes the SAF as a weighted combination of projections of the frequency vector $\tilde{\bm{\omega}}$ onto the nontrivial left singular vectors $\bm{u}^j$. However, the quadratic form in \cref{eq:02:06} has some advantages that we will leverage below in our analysis of the expected value and variance of $J(\tilde{\bm{\omega}},L)$ when considering randomized frequencies.

Together, \cref{eq:02:05} and \cref{eq:02:06} yield
\begin{align}
r\approx1-\frac{J(\tilde{\bm{\omega}},L)}{2K^2H'^2(0)},\label{eq:02:08}
\end{align}
thus allowing us to use the SAF as an objective measure for the degree of synchronization in a network. On its own, the SAF quantifies the interplay between a network's dynamics (i.e., through the natural frequencies $\{\omega_i\}$) and its structure (i.e., through Laplacian singular values $\{s_i\}$ and singular vectors $\{\bm{u}_i\}$ and $\{\bm{v}_i\}$). Equation~\cref{eq:02:08} implies that, in the linearized regime where the approximations made to obtain \cref{eq:02:03} are valid, the synchronization properties of a given network are promoted when the SAF $J(\tilde{\bm{\omega}},L)$ is made small. Specifically, by inspecting \cref{eq:02:07}, we can see that this is done by aligning the frequency vector with the singular vectors associated with large singular values and making it orthogonal to the singular vectors associated with small singular values. We also note here that in the typical case of sinusoidal coupling, $H(\theta)=\sin(\theta)$, we have that $H'(0)=1$ and $H(0)=0$, implying $\tilde{\omega}_n=\omega_n$ and simplifying \cref{eq:02:08} to 
\begin{align}
r\approx1-\frac{J(\bm{\omega},L)}{2K^2}.\label{eq:02:09}
\end{align}
In the remaining text, we will assume sinusoidal coupling, $H(\theta)=\sin(\theta)$.

\begin{figure}[htbp]
\centering
\includegraphics[width=0.53\textwidth]{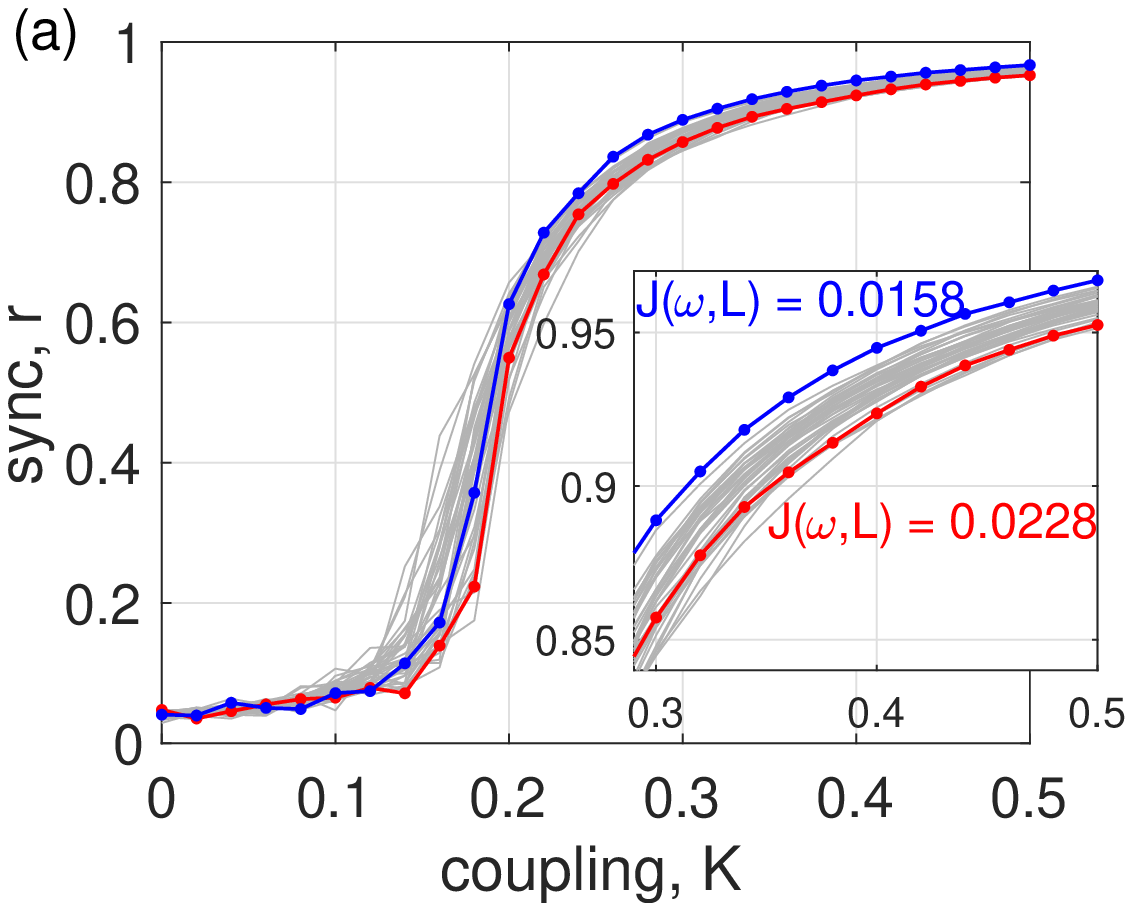}
\includegraphics[width=0.46\textwidth]{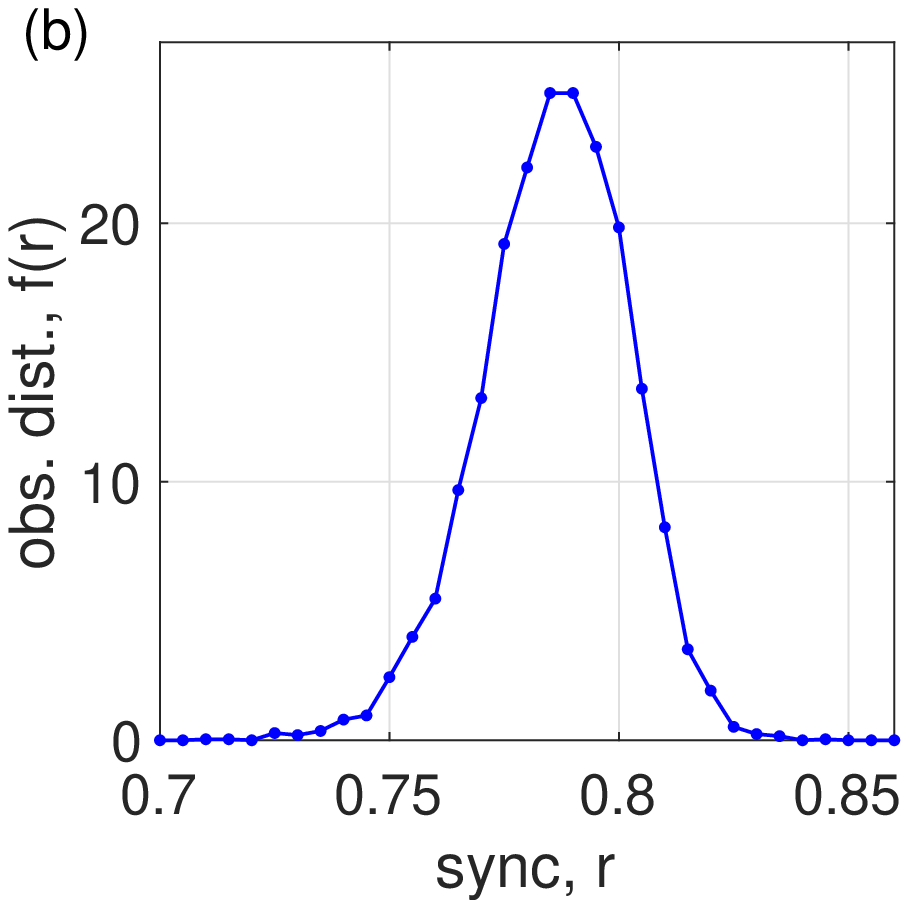}
\caption{{\bf Network Synchronization and the SAF}. (a) Order parameter $r$ versus coupling strength $K$ for the network Kuramoto phase oscillator system given by \cref{eq:02:01} with an ER network of size $N=500$ with mean degree $\langle k \rangle=10$ and $40$ different arrangements (i.e., permutations) of a set of normally distributed natural frequencies. Results are plotted in gray, except for those arrangements minimizing and maximizing $J(\bm{\omega},L)$, which are plotted in blue and red circles, respectively. (b) An empirical distribution $f(r)$ of the order parameter $r$ taken at $K=0.25$ and based on $5000$ arrangements of the same frequencies.}
\label{fig:01}
\end{figure}

\subsection{Motivating Numerical Experiment}\label{subsec:02:03}

The SAF allows us to evaluate the synchronization properties of a specific network with a specific set of natural frequencies. Moreover, \cref{eq:02:08} allows for the efficient optimization of a network's synchronization properties under various constraints. As an illustrative example, we consider a directed Erd\H{o}s-R\'{e}nyi (ER) network~\cite{Erdos1959} of size $N=500$ nodes with a mean degree $\langle k\rangle=10$ and a given set of $N$ natural frequencies $\{\omega_i\}$ drawn identically and independently from a normal distribution with zero mean and unit variance. We consider several different arrangements (i.e., permutations) of the given set of frequencies, a range of coupling strengths $K\in[0,0.5]$ and. For each arrangement and $K$ value, we numerically integrate \cref{eq:02:01} with sinusoidal coupling and a random initial condition   using  Newton's forward-difference method.

In \cref{fig:01} (a), we plot the degree of synchronization given by the order parameter $r$ as a function of the coupling strength $K$ for $40$ different arrangements of the same collection of natural frequencies on the network. To demonstrate the utility of the SAF we find the minimum and maximum value of $J(\bm{\omega},L)$ over these $40$ arrangements (given by $0.0158$ and $0.0228$, respectively), and we  plot the corresponding curves in blue and red, respectively. As predicted by \cref{eq:02:08} and \cref{eq:02:09}, the minimum and maximum values of the SAF correspond to the maximum and minimum values of the order parameter, respectively, in the strongly synchronized regime. This is highlighted in the inset where we zoom-in on the synchronized regime. Note that once synchronization occurs, the precise value of the order parameter varies significantly from one arrangement to the next. To highlight this, in Fig.~\ref{fig:01} (b) we plot an empirically observed distribution of $r$ --- which we denote $f(r)$ --- that is estimated using 5000 random arrangements at $K=0.25$, i.e., near the beginning of the strongly synchronized regime.  Despite the fact that these results come from a single network structure and a single set of natural frequencies, the distribution of observed order parameter values is remarkably wide, with a bulk ranging from roughly $0.74$ to $0.82$. (The sample mean and standard deviation are $0.7861$ and $0.0156$, respectively.)

As illustrated with this example, the SAF framework has proven very useful in both characterizing and optimizing the synchronization properties of coupled oscillator networks with specific information about the network structure and local dynamics, i.e., natural frequencies $\{\omega_i\}$. However, the requirement knowing the precise natural frequencies introduces a hurdle in quantifying and developing an understanding of the overall synchronization properties of, for instance, a network structure without specific knowledge of the natural frequencies. Instead, some statistical properties of the natural frequencies in a network are more likely attainable, in which case a description of the expected value of the SAF and the likely size of any deviations, i.e., its variance, is a more realistic outcome. Moreover, this approach lends itself better to providing a broader understanding of the network properties that promote or inhibit synchronization. To this end, in the following section we will develop theory for the expectation and variance of the SAF
by treating the natural frequencies as latent random variables.

\section{Analysis of the SAF under Uncertain Frequencies}\label{sec:03}

We will now study the behavior of the SAF in a new context where natural frequencies are not determined, but rather are random variables. We consider ensembles of natural frequency vectors and describe the expectation \drt{(\cref{sec:3:1})} and variance \drt{(\cref{sec:3:2})} of the SAF. In (\cref{sec:3.3}), we provide numerical support to validate these results.

Before continuing, we emphasize that we will still treat the network structure as fixed, and therefore the quantities that result from the network structure (e.g., the Laplacian $L$, singular vectors $\bm{u}^i$ and $\bm{v}^i$, singular values $s_i$, etc.) are also fixed. On the other hand, from this point forward, we will treat each natural frequency $\omega_i$ as a random variable with mean $\mu_i$ and variance $\sigma_i^2$. The vector $\bm{\omega}=[\omega_1,\dots,\omega_N]^T$ is then a vector of random variables with mean vector $\bm{\mu}=[\mu_1,\dots,\mu_N]^T$ and covariance matrix $\Sigma$ (whose diagonal entries are the variances $\sigma_i^2$ and whose off-diagonal entries are the covariances). Therefore, $\bm{\mu}$ and $\Sigma$ may be interpreted as latent variables for the natural frequencies $\bm{\omega}$.

\subsection{Expectation of the SAF under Uncertain Frequencies}\label{sec:3:1}

We first study the expectation $\mathbb{E}[J(\bm{\omega},L)]$ of the SAF, and we will present several results for different assumptions about $\bm{\mu}$ and $\Sigma$. Our analysis leverages the observation that the SAF given in \cref{def:02:03} is a quadratic form, allowing us to utilize known results for the quadratic forms of random vectors.

\begin{theorem}[Expectation for Quadratic Forms of Random Vectors \cite{Rencher2008}]\label{th:03:01}
Consider a random vector $\bm{y} = [y_1,\dots,y_N]^T$ in which $\bm{ \mu}=[\mu_1,\dots,\mu_N]^T$ encodes the entries' means and $\Sigma$ encodes the covariances between entries. For any real-valued square, symmetric matrix $X$, the mean of the quadratic form $\bm{y}^TX\bm{y}$ is given by
\begin{align}
\mathbb{E}\left[\bm{y}^TX\bm{y} \right]  &= \mathrm{Tr}[X\Sigma] +  \bm{ \mu}^T X \bm{ \mu} , \label{eq:03:01}
\end{align}
where $\mathrm{Tr}[\cdot]$ denotes the matrix trace.
\end{theorem}
\begin{proof}
This result is well known 
and is given as Theorem 5.2a in \cite{Rencher2008}. To provide the reader insight, we include a short proof in \cref{app:A}.
\end{proof}

We   now utilize \cref{th:03:01} to find the expectation of the SAF for a number of different choices of random frequency vectors. First we consider the SAF for general frequency vectors $\bm{\omega}$ with given mean vector $\bm{\mu}$ and covariance matrix $\Sigma$.

\begin{theorem}[Expected SAF with Random Frequencies]\label{th:03:02}
Consider the network Kuramoto phase oscillator model given in \eqref{def:Kuramoto} with linearized dynamics given by \cref{eq:02:04} for a Laplacian matrix $L$ and natural frequencies that are drawn as a random vector $\bm{\omega}$ in which the entries' means are given by $\bm{ \mu}=[\mu_1,\dots,\mu_N]^T$ and their covariances are given by $\Sigma$. The expected value of the SAF is given by
\begin{align}
\mathbb{E}[J(\bm{\omega},L)]=J(\bm{\mu},L)  + \frac{1}{N}\mathrm{Tr}[U(S^{\dagger})^2U^T \Sigma] .\label{eq:03:02}
\end{align}
\end{theorem}
\begin{proof}
Considering the quadratic form in \cref{eq:02:06}, we let $X=N^{-1}U(S^\dagger)^2U^T$ and $\bm{y}=\bm{\omega}$, which, when substituted into \cref{eq:03:01}, simplifies to \cref{eq:03:02}.
\end{proof}

The general form of the expected value of the SAF presented in \cref{th:03:02} deserves a number of remarks. Overall, the expectation of the SAF under random frequencies can be viewed as a perturbation of the SAF computed using the expected frequencies, $J(\bm{\mu},L)$. Thus, in the limit as the variance of each natural frequency approaches zero, we recover $\mathbb{E}[J(\bm{\omega},L)]=J(\bm{\mu},L)$. Next we present several corollaries of \cref{th:03:02} that consider the specific cases of natural frequency vector with equal means, independent natural frequencies with equal variances, and finally independently and identically distributed (IID) natural frequencies.

\begin{corollary}[Expected SAF for Zero-Mean Random Frequencies]\label{cor:03:03}
In the case where the natural frequencies have  zero mean, $\bm{\mu}=\bm{0}$, then the expected value of the SAF is given by
\begin{align}
\mathbb{E}[J(\bm{\omega},L)] &=\frac{1}{N}\mathrm{Tr}[U(S^{\dagger})^2U^T \Sigma].\label{eq:03:03}
\end{align}
\end{corollary}
\begin{proof}
The result is immediate  using $J({\bf 0},L)=0$.
\end{proof}

\begin{corollary}[Expected SAF for Independent Random Frequencies of Equal Variance]\label{cor:03:04}
In the case where the natural frequencies are uncorrelated and have the same variance, $\Sigma = \mathrm{Var}(\omega)I_N$, where $I_N$ is an identity matrix of size $N$, then the expected value of the SAF is given by
\begin{align}
\mathbb{E}[J(\bm{\omega},L)]=J(\bm{\mu},L) +  \frac{\mathrm{Var}(\omega)}{N} \sum_{j=2}^N \frac{1}{s_j^2},\label{eq:03:05}
\end{align}
where $s_j$ is the $j^{\text{th}}$ singular value of the Laplacian $L$.
\end{corollary}
\begin{proof}
Using the definitions of $\Sigma$ and trace, and noting that
\begin{align}
U(S^{\dagger})^2U^T=U(S^{\dagger})^TV^TVS^\dagger U^T=(L^\dagger)^TL^\dagger,\label{eq:03:06}
\end{align}
we have that
\begin{align}
\textrm{Tr}[U(S^{\dagger})^2U^T \Sigma]&=\textrm{Tr}[(L^\dagger)^TL^\dagger \Sigma]  
= \mathrm{Var}(\omega)\textrm{Tr}[(L^\dagger)^TL^\dagger]  
= \mathrm{Var}(\omega)\sum_{j=1}^N  \hat{\lambda}_j, \label{eq:03:07}
\end{align}
where we have used that the trace of a matrix equals the sum of its eigenvalues and we denote the eigenvalues of $(L^\dagger)^TL^\dagger$ as $\{\hat{\lambda}_i\}_{i=1}^N$. Since $L^\dagger$ is the pseudo-inverse of the Laplacian $L$ we have that $\hat{\lambda}_1=0$ and $\hat{\lambda}_i=s_i^{-2}$ for $i=2,\dots,N$, which completes the proof.
\end{proof}

\begin{corollary}[Expected SAF for IID Frequencies]\label{cor:03:05}
In the case where the natural frequencies are identically and independently distributed with zero mean, then the expected value of the SAF is given by
\begin{align}
\mathbb{E}[J(\bm{\omega},L)] &=\frac{\mathrm{Var}(\omega)}{N} \sum_{j=2}^N \frac{1}{s_j^2}.\label{eq:03:08}
\end{align}
\end{corollary}
\begin{proof}
The result follows directly from applying \cref{cor:03:03} and \cref{cor:03:04}.
\end{proof}

Note that because $\mathrm{Var}(\omega)>0 $ and $ \sum_{j=2}^N  {s_j^{-2}}>0$, if they are IID then the uncertainty for the natural frequencies, always increases the SAF in expectation, thereby inhibiting synchronization.
 
\subsection{Variance of the SAF for Normally Distributed Frequencies}\label{sec:3:2}

We now study the variance of the SAF, $\mathrm{Var}[J(\bm{\omega},L)]=\mathbb{E}[J^2(\bm{\omega},L)]-\mathbb{E}[J(\bm{\omega},L)]^2$. Similar to the previous subsection, we will present several results that are obtained under different assumptions about $\bm{\mu}$ and $\Sigma$. We will again leverage the quadratic form of the SAF given in \cref{def:02:03}, now with a theorem describing the variance of a quadratic forms of random vectors. In order to obtain analytical expressions, we now restrict our attention to the case of normally distributed natural frequencies. 

\begin{theorem}[Variance for Quadratic Forms of Normal-Distributed Random Vectors \cite{Rencher2008}]\label{th:03:06}
Consider the random vector $\bm{y}$ and matrix $X$ from Thm.~\ref{th:03:01}, and further assume $\bm{y}\sim \mathcal{N}(\bm{\mu},\Sigma)$ is drawn from a multivariate Gaussian distribution. Then the variance of the quadratic form $\bm{y}^TX\bm{y}$ is given by
\begin{align}
\mathrm{Var}\left[\bm{y}^TX\bm{y}\right]  &= 2 \mathrm{Tr}[ (X \Sigma )^2] +   4 \bm{ \mu}^T X\Sigma X \bm{ \mu}.\label{eq:03:09}
\end{align}
\end{theorem}
\begin{proof}
See  Theorem 5.2c of \cite{Rencher2008} for a derivation the uses the moment generating function for $\bm{y}^TX\bm{y}$. 
\end{proof}

We will now utilize \cref{th:03:06} to evaluate the variance of the SAF for a number of different choices of random frequency vectors. First we consider the SAF for general frequency vectors $\bm{\omega}$ with given mean vector $\bm{\mu}$ and covariance matrix $\Sigma$.

\begin{theorem}[Variance of SAF for Normal-Distributed Frequencies]\label{th:03:07}
Consider the network of oscillators with random frequencies described in \cref{th:03:02} and further assume the frequency vector $\bm{\omega}\sim \mathcal{N}(\bm{\mu},\Sigma)$ is drawn from a multivariate Gaussian distribution with mean $\bm{\mu}$ and variance $\Sigma$. Then the variance of the SAF $J(\bm{\omega},L)$ is given by
\begin{align}
\mathrm{Var}\left[J(\bm{\omega},L)\right] 
&= \frac{1}{N^2}\left[2 \mathrm{Tr}[ [U(S^{\dagger})^2U^T \Sigma]^2] +   4\bm{ \mu}^T U(S^{\dagger})^2U^T \Sigma U(S^{\dagger})^2U^T \bm{\mu}  \right].\label{eq:03:10}
\end{align}
\end{theorem}
\begin{proof}
Similarly as the proof to \cref{th:03:02}, we use the SAF given via a quadratic form in \cref{eq:02:06}, and let $X=N^{-1}U(S^\dagger)^2U^T$ and $\bm{y}=\bm{\omega}$. Applying \cref{th:03:07} yields the desired result.
\end{proof}

The variance of the SAF presented in \cref{th:03:07} does not have as straight forward of an interpretation as the expected value (see \cref{th:03:02}). However, it has an important similarity in that it consists of two terms, one term is a quadratic form with the mean vector $\bm{\mu}$, and the other term is a trace of a matrix product that includes the covariance matrix $\Sigma$. Next we present several corollaries of \cref{th:03:02} that, similar to the results given in section~\ref{sec:3:1}, consider the specific cases of natural frequency vector with equal means, independent natural frequencies with equal variances, and IID natural frequencies.

\begin{corollary}[Variance of SAF for Equal-Mean Random Frequencies]\label{cor:03:08}
In the case where the natural frequencies have the same mean, $\bm{\mu}=\overline{\omega}{\bf 1}$, then the variance of the SAF is given by
\begin{align}
\mathrm{Var}\left[J(\bm{\omega},L)\right] &= \frac{2}{N^2}\mathrm{Tr}[ [U(S^{\dagger})^2U^T \Sigma]^2]   .\label{eq:03:11}
\end{align}
\end{corollary}
\begin{proof}
Similarly to the proof for \cref{cor:03:03}, we may rescale the mean vector to $\bm{\mu}=\bm{0}$. The latter term on the right hand side of \cref{eq:03:10} then vanishes, i.e., 
\begin{align}
\bm{0}^T U(S^{\dagger})^2U^T \Sigma U(S^{\dagger})^2U^T \bm{0}=0,\label{eq:03:12}
\end{align}
completing the proof
\end{proof}

\begin{corollary}[Variance of SAF for Independent Random Frequencies of Equal Variance]\label{cor:03:09}
In the case where the natural frequencies are uncorrelated and have the same variance, $\Sigma = \mathrm{Var}(\omega)I_N$, where $I_N$ is an identity matrix of size $N$, then the variance of the SAF is given by
\begin{align}
\mathrm{Var}\left[J(\bm{\omega},L)\right] &= \frac{1}{N^2}\left[2 \mathrm{Var}(\omega)^2  \sum_{j=2}^N  \frac{1}{s_j^4} +   4 \mathrm{Var}(\omega)\bm{ \mu}^T U(S^{\dagger})^4U^T \bm{\mu}\right].\label{eq:03:13}
\end{align}
\end{corollary}

\begin{proof}
Applying \cref{th:03:07} and beginning with the first term on the right hand side of \cref{eq:03:10}, we have that
\begin{align}
\mathrm{Tr}[ [U(S^{\dagger})^2U^T \Sigma]^2]=\mathrm{Var}(\omega)^2\mathrm{Tr}[ (L^\dagger)^TL^\dagger (L^\dagger)^TL^\dagger]= \mathrm{Var}(\omega)^2\sum_{j=1}^N  \hat{\lambda}_j, \label{eq:03:14}
\end{align}
where we have used that $U(S^\dagger)^2U^T=(L^\dagger)^2L^\dagger$ (see \cref{eq:03:06}) and $\mathrm{Tr}[(L^\dagger)^TL^\dagger (L^\dagger)^TL^\dagger]$ is equal to the sum of its eigenvalues, denoted $\{\hat{\lambda}_j\}_{j=1}^N$. Moreover, similar to the proof for \cref{cor:03:04} we have that $\hat{\lambda}_1=0$ and $\hat{\lambda}_i=s_i^{-4}$ for $i=2,\dots,N$, completing the first term on the right hand side of \cref{eq:03:13}. For the other term, we take the second term on the right hand side of \cref{eq:03:10}, which simplifies to
\begin{align}
\bm{ \mu}^T U(S^{\dagger})^2U^T \Sigma U(S^{\dagger})^2U^T \bm{\mu}&=\mathrm{Var}(\omega)\bm{ \mu}^T U(S^{\dagger})^2U^T U(S^{\dagger})^2U^T \bm{\mu}\nonumber\\&=\mathrm{Var}(\omega)\bm{ \mu}^T U(S^{\dagger})^4U^T \bm{\mu},\label{eq:03:15}
\end{align}
completing the second term on the right hand side of \cref{eq:03:13}, completing the proof.
\end{proof}

\begin{corollary}[Variance of SAF for IID Frequencies]\label{cor:03:10}
In the case where the natural frequencies identically and independently distributed, then the variance of the SAF is given by
\begin{align}
\mathrm{Var}\left[J(\bm{\omega},L)\right] 
&= \frac{2 \mathrm{Var}(\omega)^2}{N^2}   \sum_{j=2}^N  \frac{1}{s_j^4}.\label{eq:03:16}
\end{align}
\end{corollary}
\begin{proof}
The result follows directly from applying \cref{cor:03:08} and \cref{cor:03:09}.
\end{proof}

\subsection{Numerical Validation of Theory and Definition of Link Weight Localization}\label{sec:3.3}

We conclude this section by presenting numerical experiments to illustrate the accuracy of the analytical results developed  above. In an attempt to kill two birds with one stone, we will present an experiment that both validates our theory and introduces a network property that will be explored in detail in the next section: link weight localization. Roughly speaking, weight localization describes the trade-off between sparser network structures with more strongly weighted links versus denser network structures with more weakly weighted links. A reasonable metric for measuring weight localization of a given network is the mean weight associated with each link in the network.

\begin{definition}[Weight Localization] \label{def:Localization}
Consider a weighted network consisting of $N$ nodes with non-negative adjacency matrix $A$. The weight localization, denoted $\ell$ is given by the mean link weight,
\begin{align}
\ell
=\frac{\sum_{n,m}\chi_{nm}A_{nm}}{\sum_{n,m}\chi_{nm}},\label{eq:04:01}
\end{align}
where $\chi$ is the link indicator function, i.e., $\chi_{nm}=1$ if $A_{nm}>0$ and $0$ if $A_{nm}=0$.
\end{definition}

We will use link weight localization to define an ensemble of Kuramoto phase oscillator systems, allowing us to illustrate the accuracy of our theory for a range of SAF values. We consider weighted versions of Erd\H{o}s--R\'enyi (ER) and scale-free (SF) networks constructed using the configuration model~\cite{Molloy1995}. Specifically, given a network with $N$ nodes we construct and $M=N\langle k \rangle/\ell$ undirected edges, each having weight $\ell$. (Recall that $k_i=\sum_{j}A_{ij}$ denotes the strength, or weighted degree, of node $i$ and $\langle k\rangle$ denotes the average of $k_i$.) 

\begin{figure}[htbp]
\centering
\includegraphics[width=0.45\textwidth]{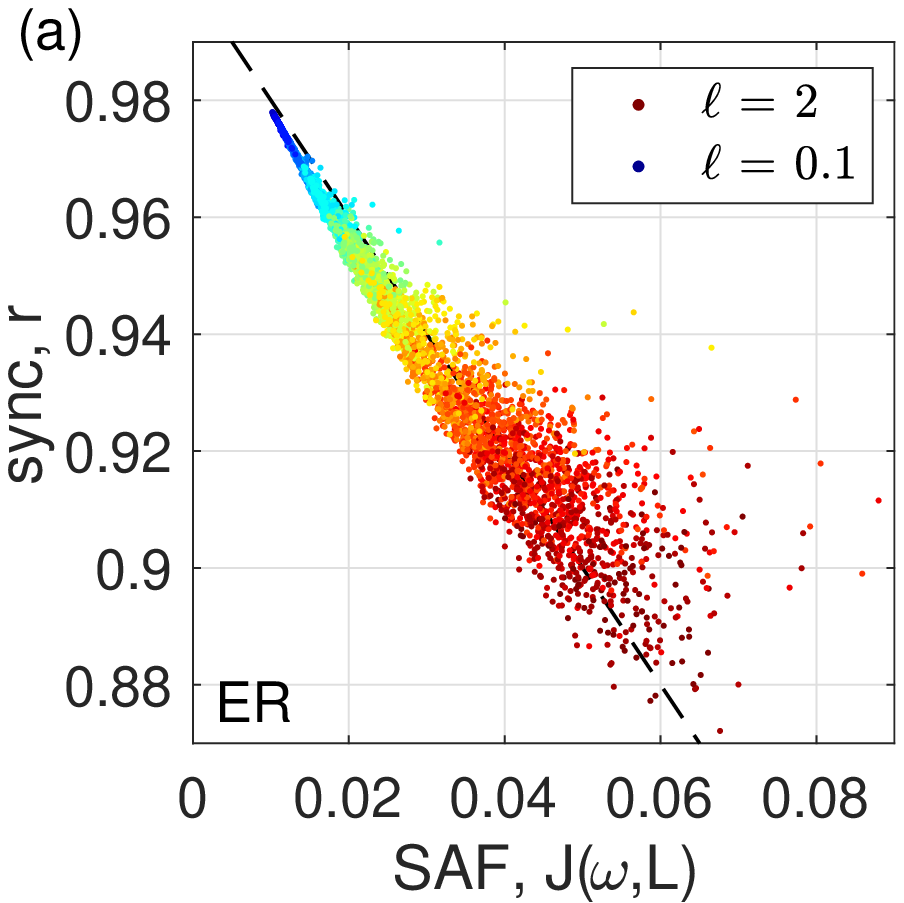}
\includegraphics[width=0.45\textwidth]{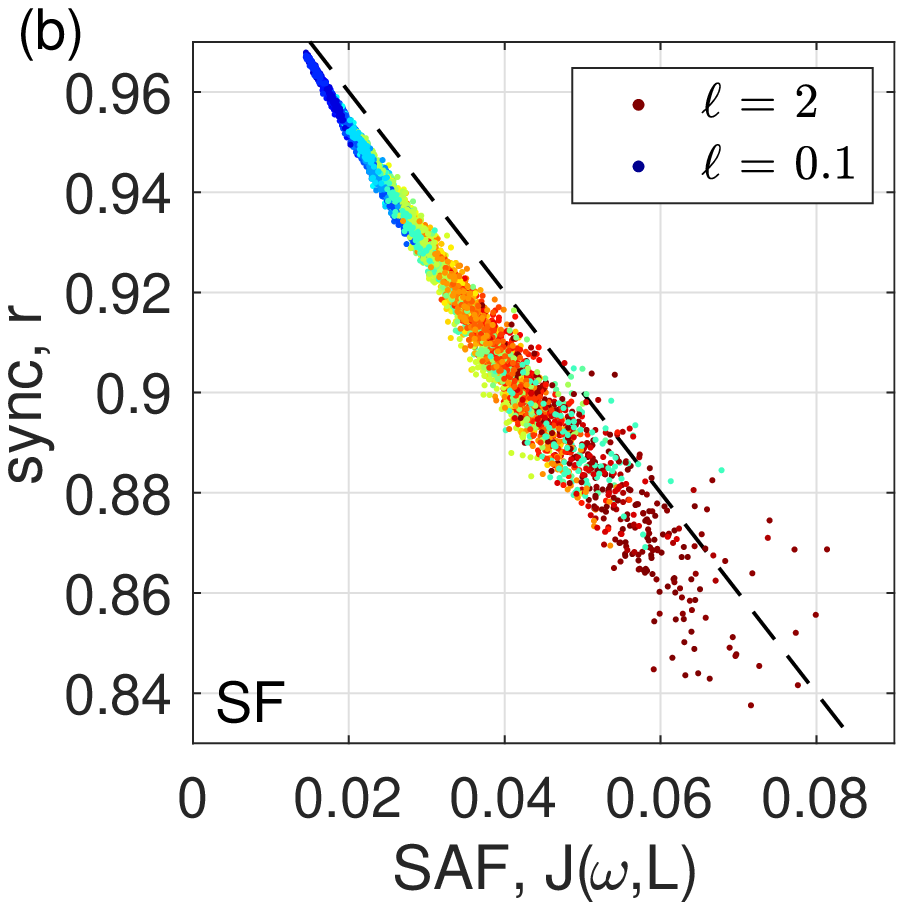}
\caption{{\bf Kuramoto order parameter $r$ versus SAF $J(\bm{\omega},L)$}. The degree of synchronization, as measured by $r$, versus the SAF at coupling strength $K=0.5$ for (a) ER and (b) SF networks with IID normal-distributed frequencies. Results are shown for $100$ networks with localization $\ell$ randomly and uniformly drawn from the interval $[0.1,2]$, and each network is paired with $50$ different random draws for the frequencies. Data points are colored according to localization, ranging from $\ell=0.1$ (blue) to $2$ (red). The linear approximation $r\approx 1-J(\bm{\omega},L)/(2K^2)$ given by \cref{eq:02:09} is indicated with the  dashed line.}
\label{fig:02}
\end{figure}

We use this class of networks to illustrate the utility of the SAF in predicting the Kuramoto order parameter $r$ in the strong synchronization regime and verify the analytical results from the previous section. We consider $100$ weighted ER and SF networks each, all of size $N=500$ with mean degree $\langle k\rangle=10$ (and $\gamma=3$ for the SF networks). For each network we chose $\ell$ randomly and uniformly from the interval $[0.1,2]$ and consider $50$ different frequency vectors with IID normal-distributed entries with unit variance and zero mean. For this choice of frequencies, the expectation and variance of the SAF are described by \cref{cor:03:05} and \cref{cor:03:10}, respectively. In \cref{fig:02}, we plot the time average of $r$ for $K=0.5$ versus the SAF $J(\bm{\omega},L)$ for  each frequency vector and each network. The dashed line indicates the linear approximation $r\approx 1-J(\bm{\omega},L)/(2K^2)$ given by \cref{eq:02:09}. Each data point is colored according to $\ell$. For both ER and SF networks, we note strong agreement between the predicted and observed values of $r$.

\begin{figure}[htbp]
\centering
\includegraphics[width=0.45\textwidth]{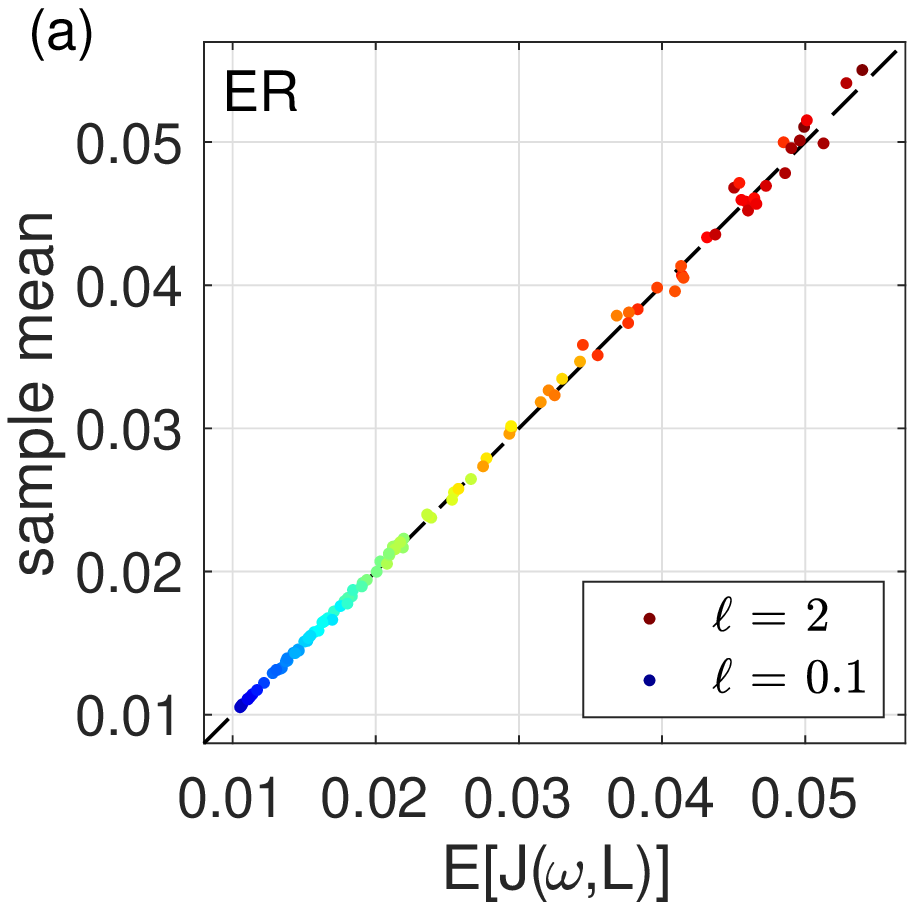}
\includegraphics[width=0.45\textwidth]{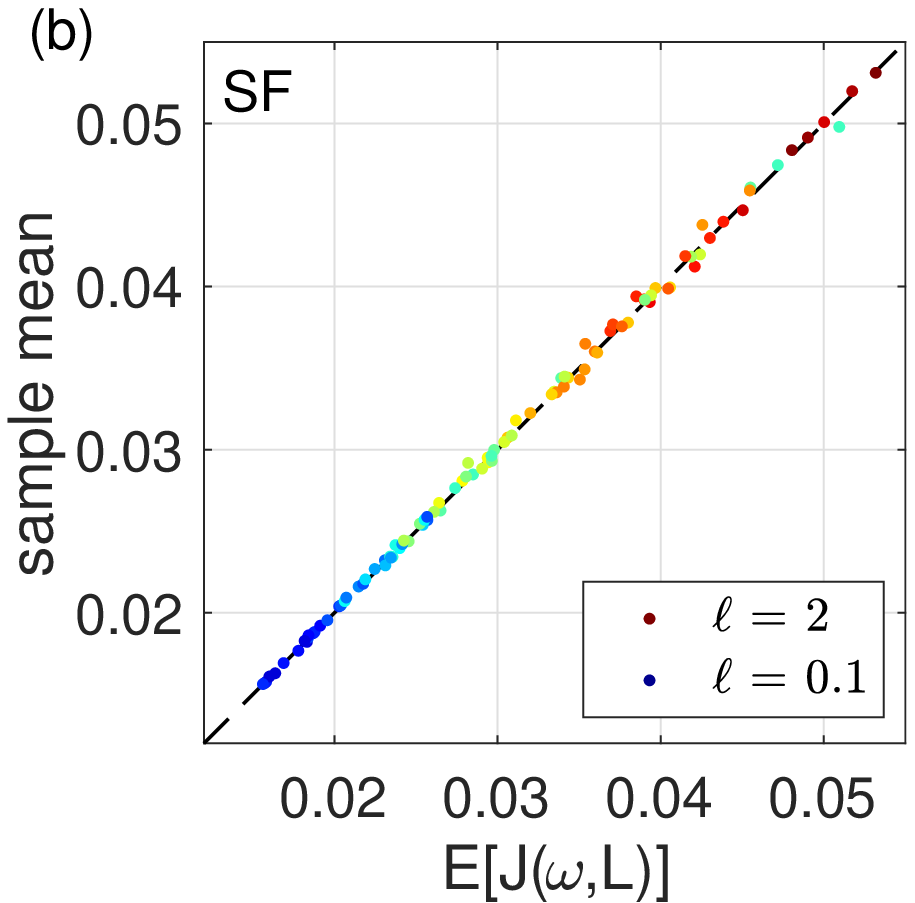}\\
\includegraphics[width=0.45\textwidth]{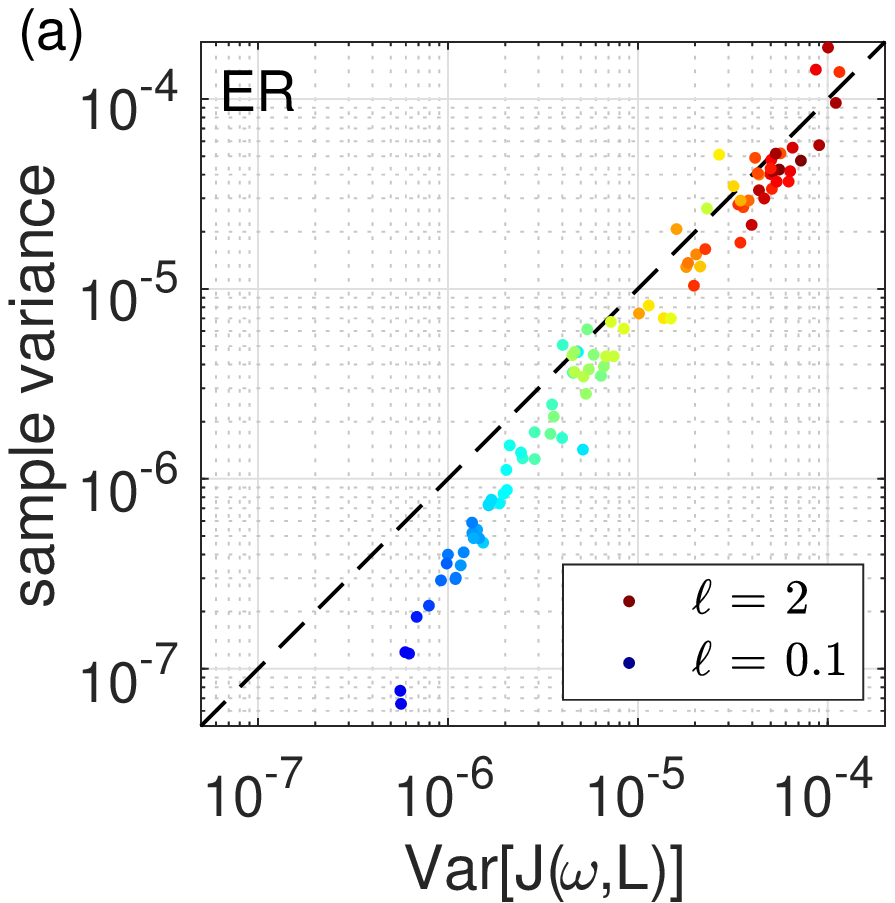}
\includegraphics[width=0.45\textwidth]{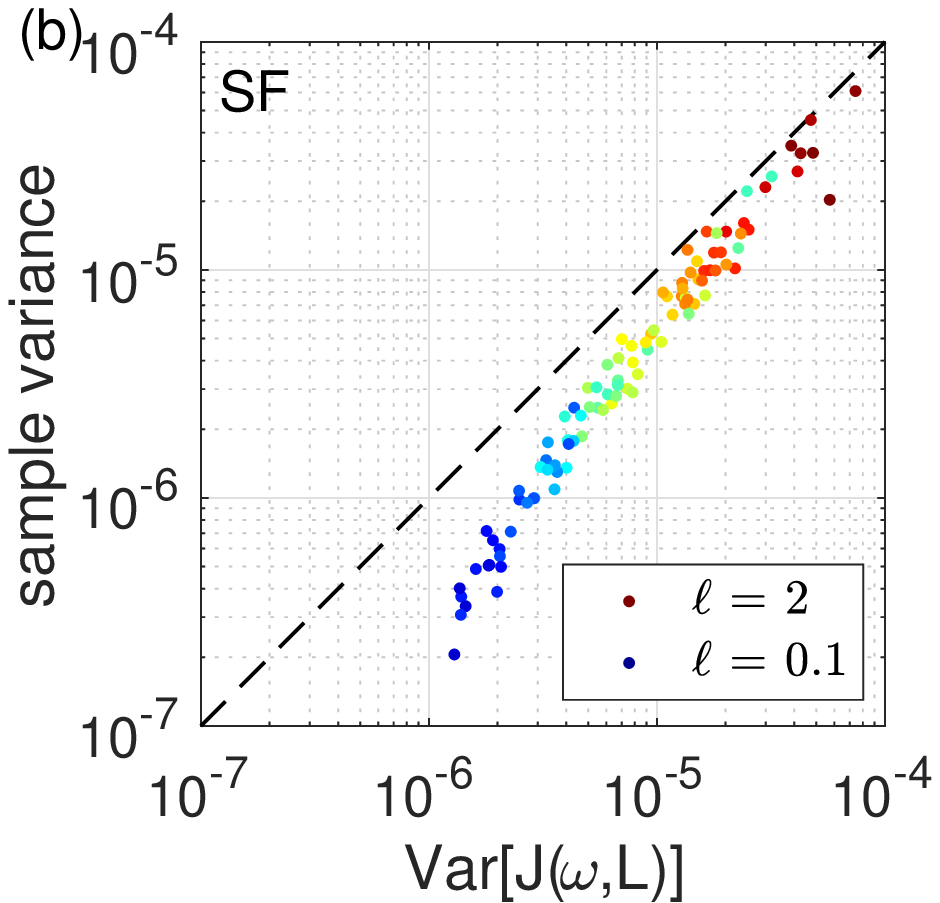}
\caption{{\bf Expectation and variance of the SAF}. We compare our predictions for $\mathbb{E}[J(\bm{\omega},L)]$ and $\mathrm{Var}[J(\bm{\omega},L)]$ given by \cref{eq:03:08} and \cref{eq:03:16} to sample estimates for the expectation and variance of the SAF. The sample estimates are based on 50 samples of IID normal-distributed frequencies for each network. Results are shown for (left) ER and (right) SF networks with edges having weight $\ell\in[0.1,2]$. 
}\label{fig:03}
\end{figure}

Next, we verify our analysis for the expectation and variance of the SAF by comparing these predictions to sample means and sample variances that are estimated using the $50$ frequency vectors for each network. In \cref{fig:03} we plot these sample means and variances versus their expected SAF and its variances, which are given by \cref{eq:03:08} and \cref{eq:03:16}, respectively. Again, we color each datapoint according to  $\ell$. Overall, we note excellent agreement between our analytical predictions and the sample means and variances observed. We note a small deviation in the variance for significantly delocalized network, where the analytical results over-predict our observations.

\section{SAF Expectation and Variance Reveals Generic Network Properties that Promote Synchronization}\label{sec:04}

In this section, we illustrate that the theoretical results obtained in the previous section provide a framework for understanding how generic network properties affect synchronization. In particular, we will study link weight delocalization (\cref{subsec:04:01}), link directedness (\cref{subsec:04:02}), and degree-frequency correlations (\cref{subsec:04:03}), and we show that each of these properties  tends to promote the synchronization properties of network-coupled heterogeneous oscillators.

\subsection{Weight Delocalization Promotes Synchronization}
\label{subsec:04:01}

The first network property that we identify as a promoter of network synchronization is weight delocalization. Recall from \cref{sec:3.3} that weight localization/delocalization quantifies the trade-off between sparse/dense structures and strong/weak connections. Specifically, in the context of conserving the total of all link weights in the network, one may strengthen the average link weight in the network at the expense of ending up with fewer links (and therefore a sparser structure) or one may increase the number of total links in the network (therefore ending up with a denser structure) at the expense of weakening the average link weight. At first glance the effect that this trade-off has on overall synchronization dynamics is not clear, however we will show that weight delocalization improves synchronization while weight localization diminishes it. 

We begin with some numerical simulations by plotting in \cref{fig:04} the Kuramoto order parameter $r$ versus coupling strength $K$ for networks with different levels of localization ranging from $\ell=2$ (localized; red, bottom) to $\ell=0.25$ (delocalized; blue, top). Similar to the simulation in the previous section, we consider IID normal-distributed natural frequencies $\{\omega_i\}$ and (a) ER and (b) SF networks of size $N=500$ with mean degree $\langle k\rangle=10$ with exponent $\gamma=3$ for SF networks. Each datapoint represents an average of $r$ over $50$ networks. Importantly,   note that in both ER and SF networks the delocalized networks (i.e., small $\ell$) display stronger synchronization properties than the localized networks (i.e., larger $\ell$). That is, synchronization is promoted by having more links with weaker link weights.

\begin{figure}[htbp]
\centering
\includegraphics[width=0.49\textwidth]{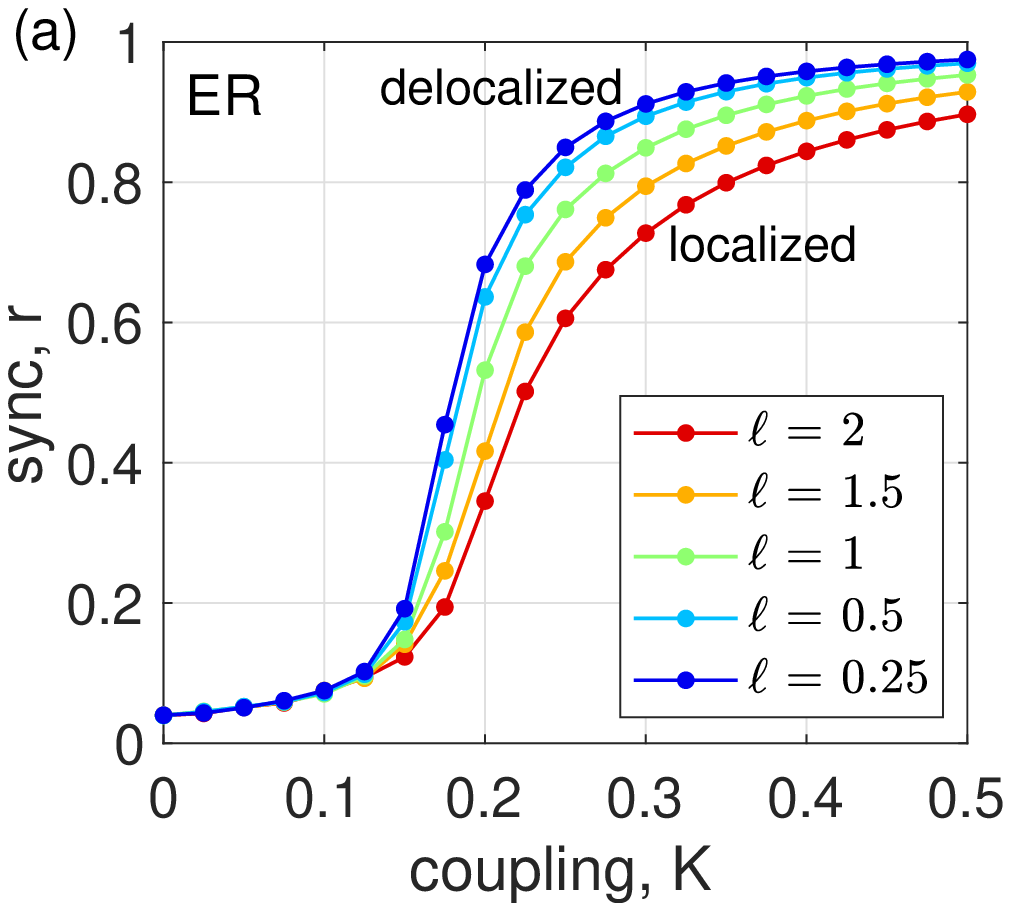}
\includegraphics[width=0.49\textwidth]{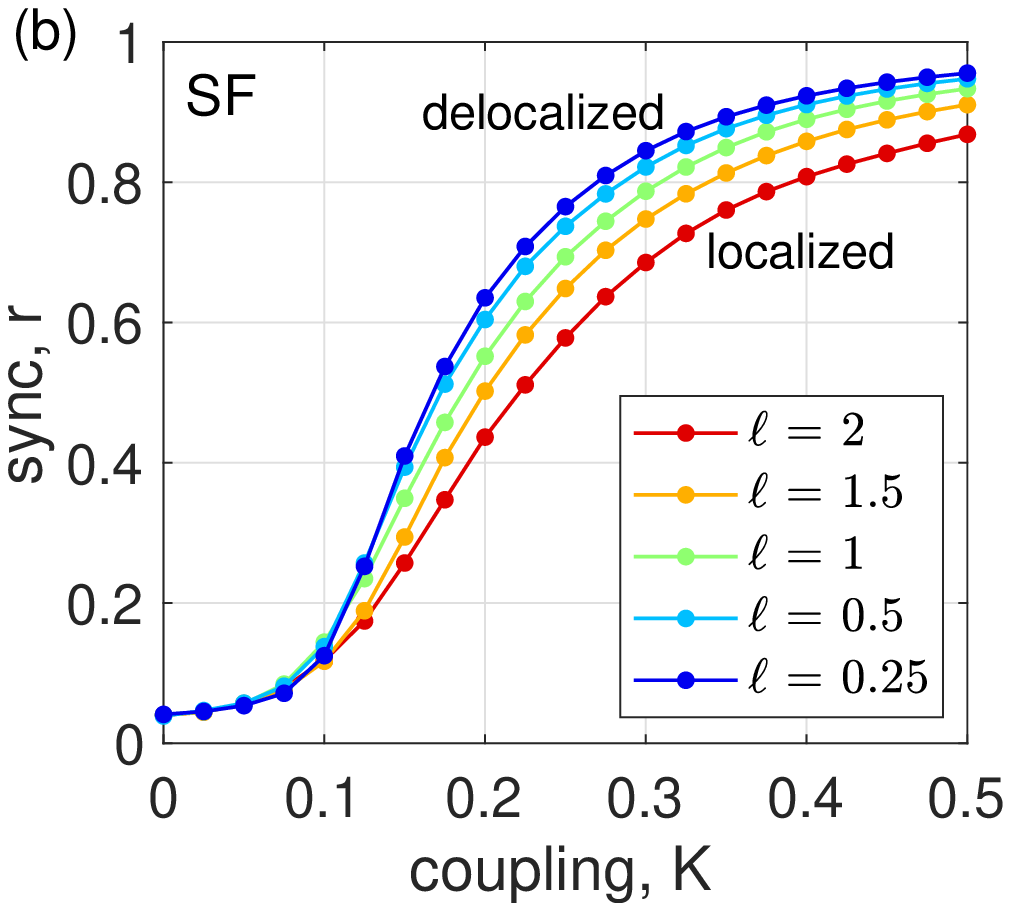}
\caption{{\bf Weight delocalization promotes synchronization.} The Kuramoto order parameter $r$ versus the coupling strength $K$ for weighted versions of (a) ER and (b) SF networks. Weight localization varies from $l=0.25$ (blue, top) to $2$ (red, bottom). All networks are of size $N=500$ nodes with mean degree $\langle k \rangle=10$ (and degree exponent $\gamma=3$ for the SF networks). Data points represent an average over $50$ networks. The oscillators' natural frequencies are drawn from the standard normal distribution.}\label{fig:04}
\end{figure}

To  explain this phenomenon, we first apply and extend some results from random matrix theory to understand the effect that weight delocalization has on the eigenvalue spectrum of the network's Laplacian matrix. Note first that the results in \cref{fig:04} used IID natural frequencies, and therefore the analytical description of the expected value and variance of the SAF were described by \cref{cor:03:05} and \cref{cor:03:10}. Moreover, in the context of undirected networks where the Laplacian $L$ is symmetric, i.e., $L^T=L$, the singular value decomposition reduces to the eigenvalue/eigenvector diagonalization of $L$. In particular, the singular values $\{s_i\}_{i=1}^N$ are given precisely by the eigenvalues $\{\lambda_i\}_{i=1}^N$. Therefore, to understanding the effect of weight delocalization on the expectations of the SAF, and thereby the overall synchronization properties of a network, we must understand the effect of weight delocalization on the spectral density  $\rho(\lambda)$ of Laplacian eigenvalues of a network. Note that once the behavior of $\rho(\lambda)$ is well understood, both $\mathbb{E}[J(\bm{\omega},L)]$ and $\mathrm{Var}[J(\bm{\omega},L)]$ can be easily predicted since they are proportional to the variance of frequencies---which we  denote $\mathrm{Var}(\omega)$---as well as the second and fourth moments of $\rho(\lambda)$ given by $\sum_{j=2}^N\lambda_j^{-2}$ and $\sum_{j=2}^N\lambda_j^{-4}$, respectively. [See \cref{eq:03:08} and \cref{eq:03:16}.]

We first study the effect of localization on the quantities $\sum_{j=2}^N\lambda_j^{-2}$ and $\sum_{j=2}^N\lambda_j^{-4}$ for $k$-regular networks, which have the property that all nodes have precisely the same degree, $k_n=k$. Moreover, we begin with unweighted case, where each of a node's $k$ links have weight $w=1$. In this case McKay's theorem characterizes the eigenvalue spectrum of the Laplacian matrix for the asymptotic, large-$N$ limit.

\begin{theorem}[McKay's Theorem \cite{McKay1981}]\label{th:04:02}
Consider a sequence of unweighted $k$-regular networks with Laplacian matrices $L^{(m)}$, each of size $N_m$ with $N_m\to\infty$ as $m\to\infty$ . Further, let $\{\lambda_i^{(m)}\}_{i=1}^{N_m}$ denote the $N_m$ eigenvalues for each $L^{(n)}$. Then the empirical spectral density $\rho_m(\lambda)= \frac{1}{N_m}\sum_{i=1}^{N_m} \delta_{\lambda_i} $ converges in probability with $m\to\infty$ as 
\begin{align}
\rho_m(\lambda) \to \rho(\lambda)=
\begin{cases}
\frac{\displaystyle k \sqrt{4 (k -1)  - (\lambda-k)^2}}{\displaystyle  2 \pi (k^2 - (\lambda-k)^2)} ,& \text{if $|\lambda-k| \leq 2
  \sqrt{k-1},$}\\ 
0 ,& \text{otherwise.}
\end{cases}\label{eq:04:02}
\end{align}
\end{theorem}

McKay's Theorem is a pioneering result in random matrix theory and research into the behavior of spectral densities for various matrices. The density $\rho(\lambda)$, which is centered at the node degree $k=\int_\lambda \lambda \hat{\rho}(\lambda)d\lambda$ and has bounded support on the interval $\lambda \in (k- 2\sqrt{k-1} ,k+ 2\sqrt{k-1} )$. 
We note that there remains an active research community working to extend McKay's law. In particular, it was recently shown that the same result can be obtained by allowing $k$ to increase with $n$, provided that $k$ grows sufficiently slowly \cite{Dumitriu2012,Tran2013}, in which case $\rho(\lambda)$ converges to the famous  ``semi-circle distribution''~\cite{Mehta2004}.

With \cref{th:04:02} in hand, we turn our attention to weighted $k$-regular networks. Given a weight $\ell$ a $k$-regular network consists of $N$ nodes that each have $k/\ell$ links with weight $\ell$. (Note that we must choose $\ell$ so that $k/\ell$ is a positive integer less than $N$.) We now a new theorem that is akin to McKay's Thereom but is modified for weighted $k$-regular networks.

\begin{theorem}[Spectral Density of Weighted $k$-Regular Networks]\label{th:04:03}
Consider a sequence of weighted $k$-regular networks with weight $\ell$ with Laplacian matrices $L^{(m)}$, each of size $N_m$ with $N_m\to\infty$ as $m\to\infty$. Further, let $\{\lambda_i^{(m)}\}_{i=1}^{N_m}$ denote the $N_m$ eigenvalues for each $L^{(n)}$. Then the empirical spectral density $\rho_m(\lambda)= \frac{1}{N_m}\sum_{i=1}^{N_m} \delta_{\lambda_i} $ converges in probability with $m\to\infty$ as 
\begin{align}
\rho_m(\lambda) \to \rho(\lambda)=
\begin{cases}
\frac{\displaystyle k \sqrt{4 (\ell k -\ell^2)  - (\lambda-k)^2}}{\displaystyle  2 \pi \ell\left[k^2 - (\lambda-k)^2\right]} ,& \text{if $|\lambda-k| \leq 2\sqrt{\ell k-\ell^2},$}\\ 
0 ,& \text{otherwise,}
\end{cases}\label{eq:04:03}
\end{align}
\end{theorem}
\begin{proof}
See appendix \ref{app:B}.

\end{proof}

A key characteristic of the spectral density $\rho(\lambda)$ is its support and heterogeneity. Specifically, its support is the interval $(k-2\sqrt{\ell k-\ell^2},k+2\sqrt{\ell k-\ell^2})$. Thus, as the weight $\ell$ increases the support widens, yielding more heterogeneous eigenvalues, and as $\ell$ decreases the support decreases, yielding more homogeneous eigenvalues (i.e., spectral concentration). In \cref{fig:05}, we plot observed spectral densities for several choices of $\ell \in\{0.2,0.3,0.5,1,2\}$ (colored blue to red) using $100$ networks of size $N=500$ with degree $k=10$ for each value of the weight. Symbols represent the overall observations and theoretical approximations from \cref{eq:04:03} are plotted as solid curves. We note an excellent agreement between them. 

\begin{figure}[htbp]
\centering
\includegraphics[width=0.55\textwidth]{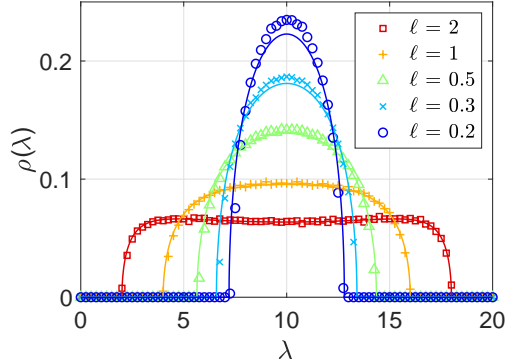}
\caption{{\bf Delocalization yields spectral concentration and promotes synchronization.} (a) Predicted spectral density $\rho(\lambda)$ given by \cref{eq:04:03} for weighted $k$-regular networks with $\ell\in[0.2,2]$. Symbols represent an empirical spectral density that is estimated using $100$ networks with $N=500$ and $k=10$. Note that decreasing $\ell$ causes the distributions to become more peaked (i.e., concentrated) around the mean $k=10$, which decreases the number of eigenvalues $\lambda_k$ close to zero and increases the SAF (i.e., since the terms $\lambda_k^{-2}$ and $\lambda_k^{-4}$ diverge as $\lambda_k\to0$).}\label{fig:05}
\end{figure}

Importantly, given an approximation of the spectral density, we can approximate the expected value and the variance of the SAF, given in \cref{eq:03:08} and \cref{eq:03:16}. For sufficiently large $N$ the expected value of the SAF can be approximated by
\begin{align}
\mathbb{E}[J(\bm{\omega},L)]&=\mathrm{Var}(\omega)\left(\frac{1}{N}\sum_{j=2}^N\frac{1}{\lambda_j^2}\right)\approx\mathrm{Var}(\omega)\int_{k-2\sqrt{\ell k-\ell^2}}^{k+2\sqrt{\ell k-\ell^2}}\frac{\rho(\lambda)}{\lambda^2}d\lambda.\label{eq:04:09}
\end{align}
On the other hand, the variance of the SAF can be approximated by
\begin{align}
\mathrm{Var}[J(\bm{\omega},L)]&=\frac{2\mathrm{Var}(\omega)^2}{N}\left(\frac{1}{N}\sum_{j=2}^N\frac{1}{\lambda_j^4}\right)\approx\frac{2\mathrm{Var}(\omega)^2}{N}\int_{k-2\sqrt{\ell k-\ell^2}}^{k+2\sqrt{\ell k-\ell^2}}\frac{\rho(\lambda)}{\lambda^4}d\lambda.\label{eq:04:10}
\end{align}
In \cref{fig:06} (a) and (b), we plot the expected value and variance, respectively, of the SAF versus localization $\ell$ for each of the $100$ networks used to create \cref{fig:05}. Observations from individual networks are plotted in circles and the dashed curves represent the predictions given by
\cref{eq:04:09} and \cref{eq:04:10}, respectively

\begin{figure}[htbp]
\centering
\includegraphics[width=0.45\textwidth]{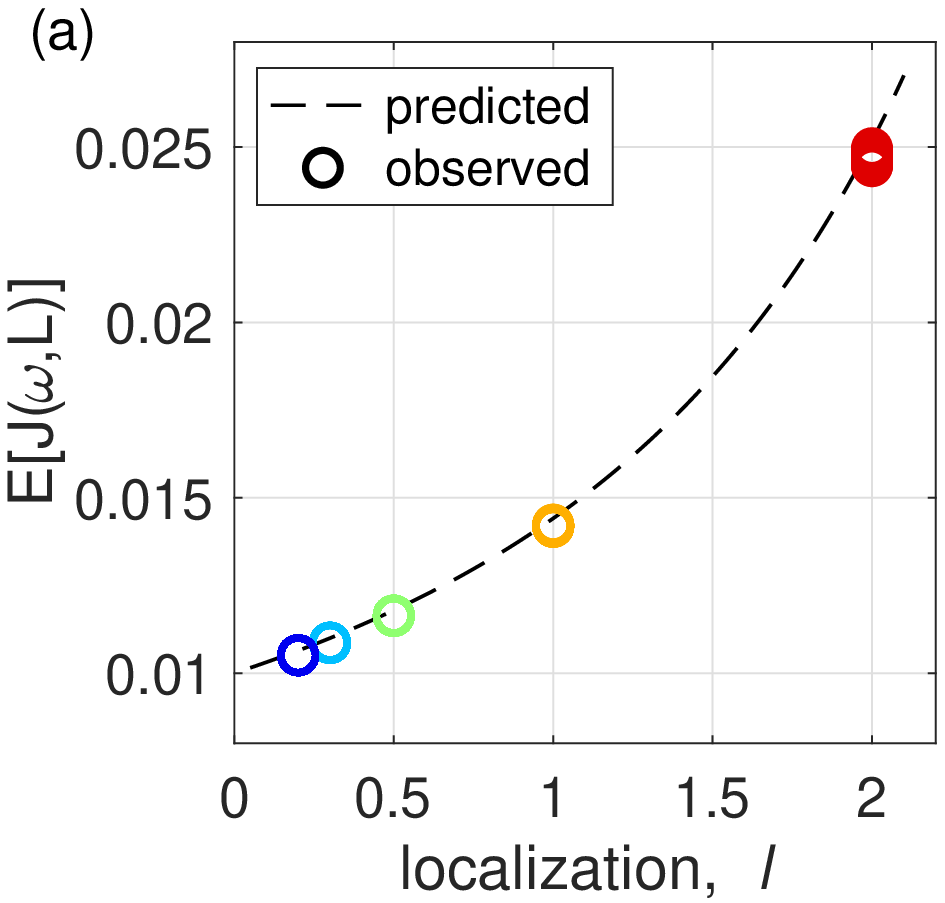}
\includegraphics[width=0.45\textwidth]{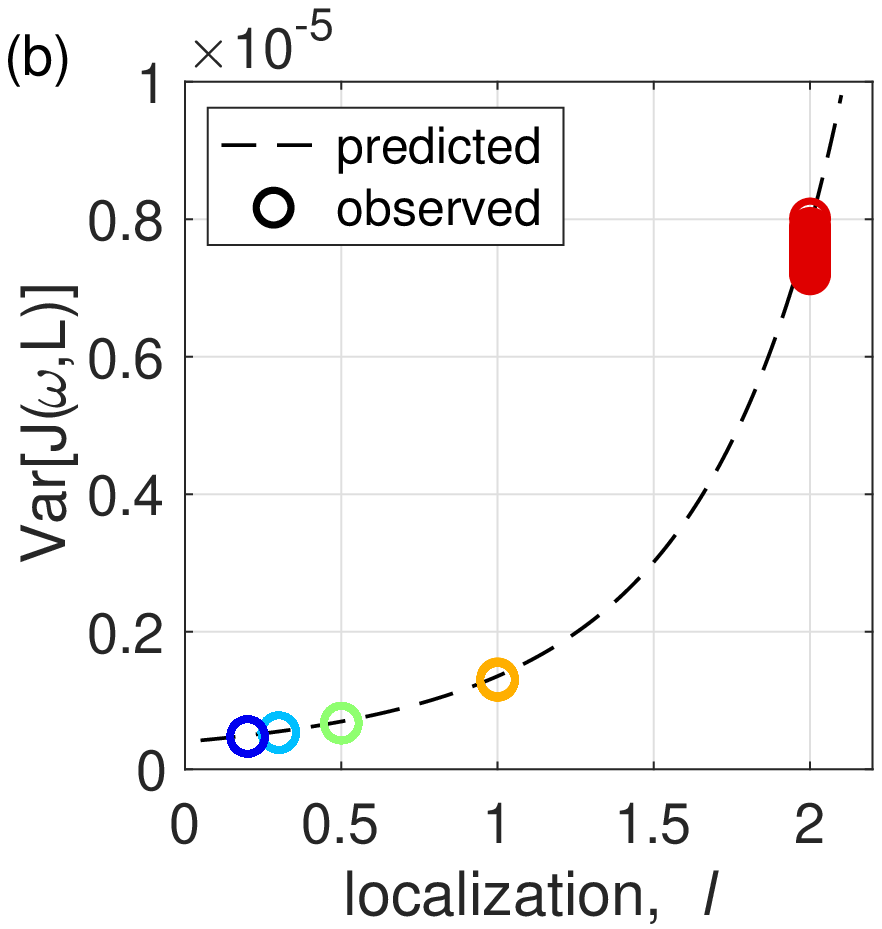}
\caption{{\bf SAF Expectation and variance of $k$-regular networks.} 
(a) Expectation and (b) variance of the SAF for $k$-regular networks with localization $\ell\in[0.2,2]$ and IID normal-distributed frequencies. Dashed curves in (a) and (b) indicate \cref{eq:04:09} and \cref{eq:04:10}. Symbols indicate empirically observed values from a collection of $100$ networks.}\label{fig:06}
\end{figure}

These results deserve a few remarks. First, we are ultimately interested in the behavior of the expected value and variance of the SAF via the sums $\sum_{j\ge2}\lambda_j^{-2}$ and $\sum_{j\ge2}\lambda_j^{-4}$, respectively. Second, we note that as network weights become more delocalized the spectral density becomes more homogeneous (i.e., concentrated), with eigenvalues clustering more tightly around the mean degree $k$. This yields smaller values of $\mathbb{E}[J(\bm{\omega},L)]$ and $\mathrm{Var}[J(\bm{\omega},L)]$, i.e., better and less variable synchronization properties of a network. Also note that, because the degree $k$ remains constant, even as localization is varied, the mean of the eigenvalue distribution remains constant since the sum of the eigenvalues equals the trace of the Laplacian, which is in turn equal to the sum of the degrees, $Nk$. Thus, as the network links become more delocalized the smaller eigenvalues increase away from zero, contributing smaller values to the sums $\sum_{j\ge2}\lambda_j^{-2}$ and $\sum_{j\ge2}\lambda_j^{-4}$, thereby improving synchronization. On the other hand, as the network links become more localized the smaller eigenvalues move closer to zero, contributing larger values to the sums $\sum_{j\ge2}\lambda_j^{-2}$ and $\sum_{j\ge2}\lambda_j^{-4}$, thereby increasing the SAF and inhibiting synchronization.

\begin{figure}[htbp]
\centering
\includegraphics[width=0.49\textwidth]{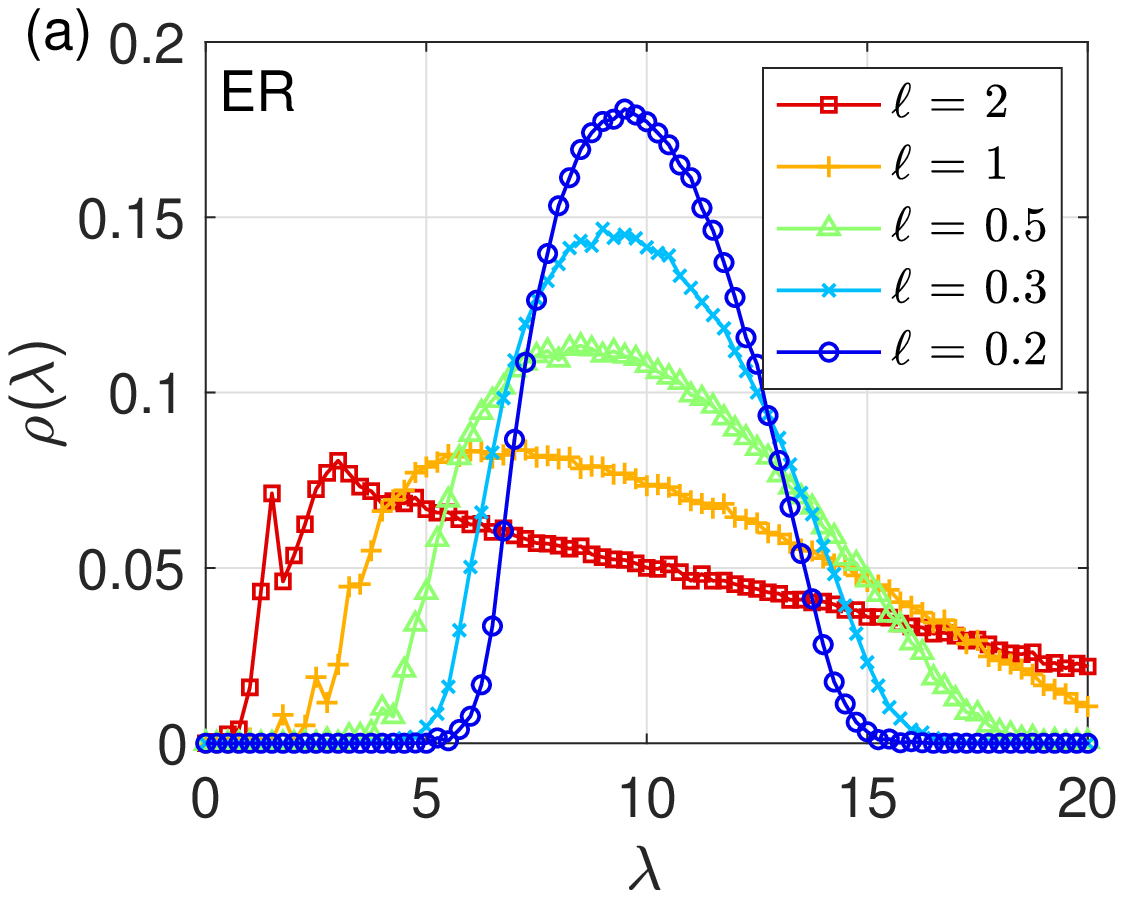}
\includegraphics[width=0.49\textwidth]{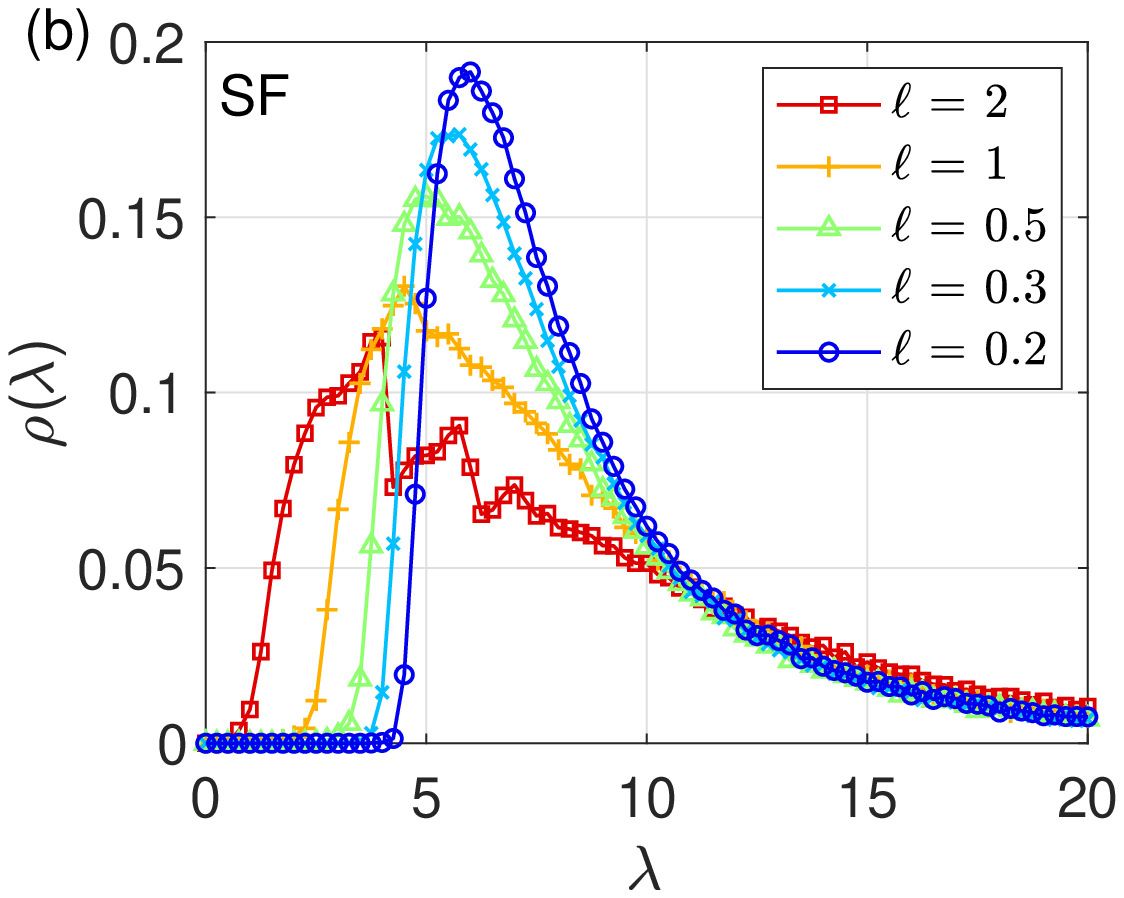}
\caption{{\bf Laplacian spectral densities for ER and SF networks}. Empirically observed spectral densities $\rho(\lambda)$ of unnormalized Laplacian matrices for (a) ER and (b) SF networks with varying localization $\ell$. Each curve represents an average across $100$ networks.} \label{fig:07}
\end{figure}

We now return to ER and SF networks, showing that a similar phenomenon occurs whereby delocalization yields a concentration of the spectral density, which ultimately decreases the SAF. We illustrate this in \cref{fig:07} (a) and (b), where we plot the spectral densities observed from collections of $100$ ER and SF networks, respectively, for several localization values ranging from $\ell=0.2$ (blue circles) to $2$ (red squares). Even for the case of ER networks, which are weakly heterogeneous compared to SF networks, it is clear that the spectral densities differ significantly from those of $k$-regular networks plotted in \cref{fig:05}. However, the same qualitative behavior persists in the context of varying weight localization as with the case of the $k$-regular networks. In particular, weight delocalization (localization) results in spectral densities that are more (less) clustered about the mean, which remains constant. Therefore, weight delocalization increases the smallest eigenvalues, resulting in both smaller expected values and variances of the SAF, thereby promoting synchronization.

\subsection{Network Directedness Promotes Synchronization}
\label{subsec:04:02}

In our analysis above we have shown that weight delocalization tends to improve the synchronization properties of networks, i.e., networks with more, weaker connections synchronize heterogeneous oscillators better than networks with fewer, stronger connections. Another fair interpretation of this phenomenon is that, to improve a network's synchronization properties, one should seek to connect more distinct pairs of oscillators, even at the expense of making the connections weaker. We now turn to another network property that can be used to attain a similar effect, namely network directedness, which we define as follows.

\begin{definition}[Network Directedness]\label{def:Directedness}
Consider a weighted network consisting of $N$ nodes with non-negative adjacency matrix $A$. The directedness of the network, denoted $p_{\text{dir}}$ is given by fraction of weighted links $m\to n$ that do not have an equal and opposite link $n\to m$,
\begin{align}
p_{\text{dir}}=\frac{1}{2N\langle k\rangle}\sum_{n,m}\left|A_{nm}-A_{mn}\right|.\label{eq:04:11}
\end{align}
\end{definition}

The directedness $p_{\text{dir}}$ can also be interpreted as the normalized sum over all entries of the absolute value of the matrix $A-A^T$. Thus, in the case of an undirected network we have that $A-A^T$ is identically zero, yielding $p_{\text{dir}}=0$. On the other hand, when each and every link in the network has no opposite counterpart, i.e., each non-zero entry of $A$ corresponds to to a zero entry of $A^T$, we have that $p_{\text{dir}}=1$ after accounting for the normalization by $2N\langle k \rangle$. Intermediate scenarios, where some links do and some links do not have an opposite counterpart, intermediate values of $p_{\text{dir}}$ are taken.

\begin{figure}[htbp]
\centering
\includegraphics[width=0.9\textwidth]{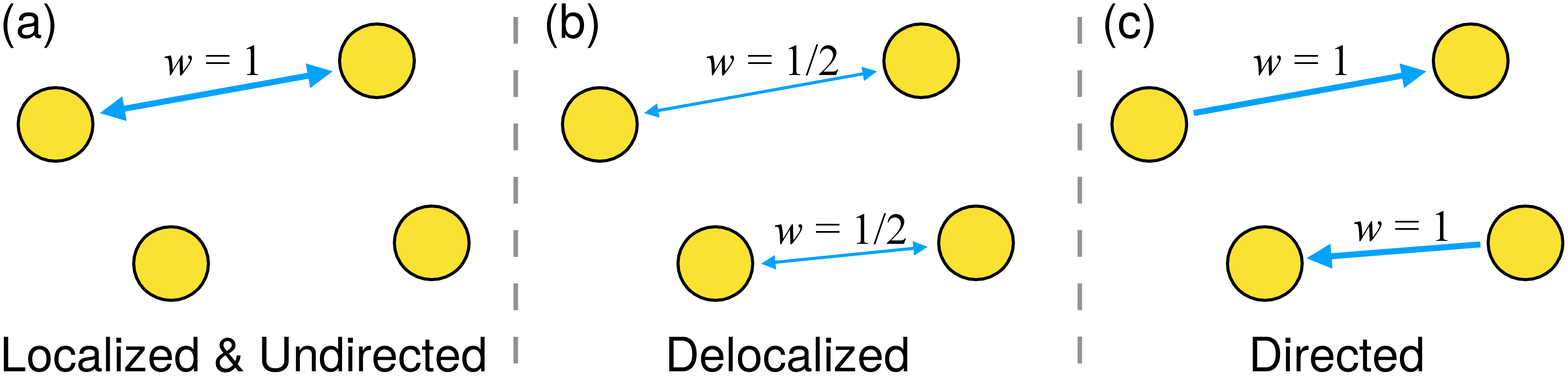}
\caption{{\bf Directedness and delocalization: Illustration}. (a) A toy network consisting of four nodes with one undirected link having weight 1 as well as (b) delocalized and (c) directed alternatives that are obtained by replacing the single undirected link with two links while conserving the networks' total link weight, $\sum_{ij}A_{ij}$ (which is 2 in each case).}\label{fig:08}
\end{figure}

In the context of studying network structures having the same size and degree (i.e., the same total link weight $\sum_{ij}A_{ij}$), increasing the directedness of a network increases the total number of pairs of nodes in the network that have a connection between them, similar to the process of delocalization. This is illustrated in \cref{fig:08}, where four nodes with just one undirected link with weight $\ell=1$ are shown in panel (a). In panel (b), we illustrate a delocalized alternative, where the set of four nodes instead share two undirected links, now with with weight $\ell=1/2$. Note that this coincides with two, not one, distinct pairs of connected nodes, while the total number of weighted links is the same as in panel (a). In panel (c), we illustrate the directed alternative, where the four nodes now share two binary directed links, again yielding two distinct pairs of connected nodes as in panel (b). We note also that in certain practical scenarios networks structures may be required to be binary, making delocalization unfeasible. Tuning directedness, however provides a more realistic alternative given the ability to connect more pairs of nodes with unweighted directed links.

\begin{figure}[htbp]
\centering
\includegraphics[width=0.55\textwidth]{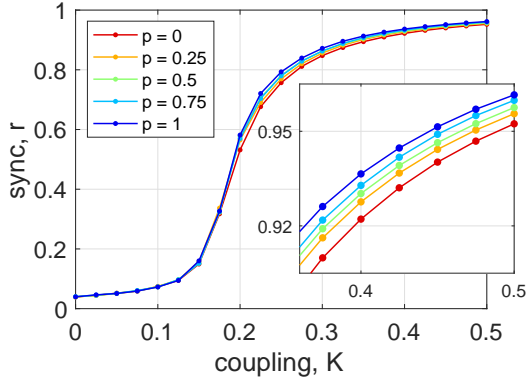}
\caption{{\bf Effect of directedness on synchronization}. The Kuramoto order parameter $r$ versus coupling strength $K$ in unweighted ER networks with directedness varying from $p_{\text{dir}}=0$ (red) to $1$ (blue). All networks are of size $N=500$ with mean degree $\langle k \rangle=10$ and $\gamma=3$ for SF, and data points represent an average over $50$ network realization. The frequencies are drawn as IID normal-distributed random variables. Inset: a zoomed-in view of the synchronization profile in the strongly synchronized regime.}\label{fig:09}
\end{figure}

\begin{figure}[htbp]
\centering
\includegraphics[width=0.49\textwidth]{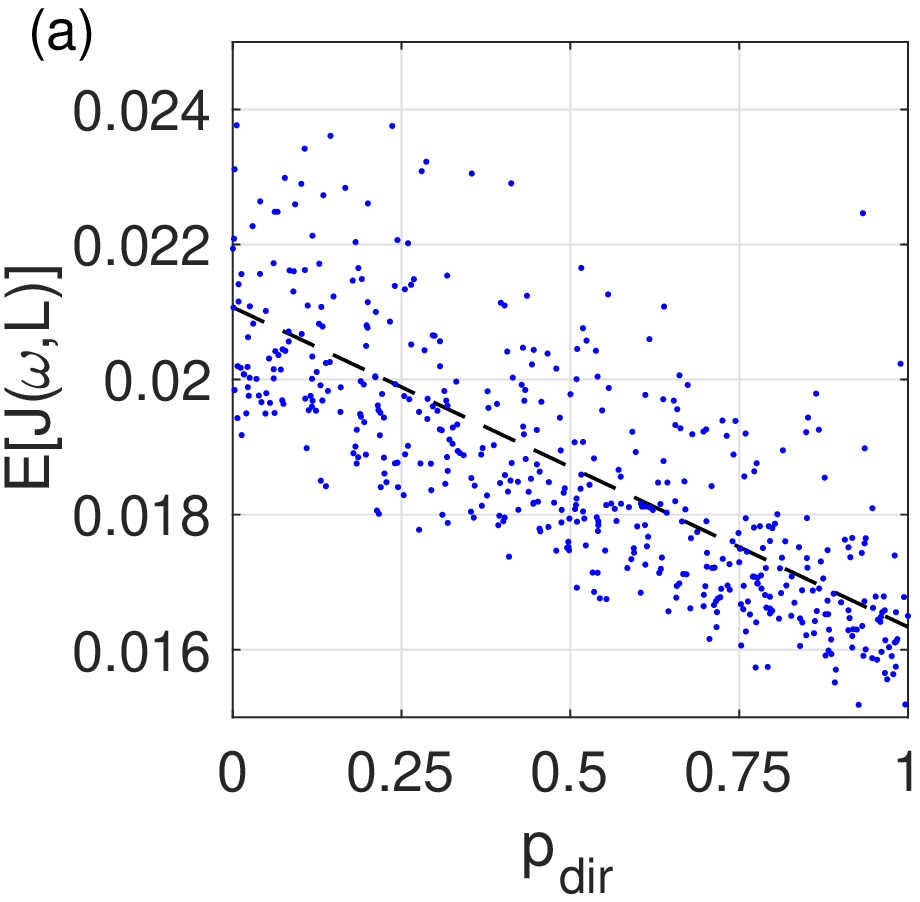}
\includegraphics[width=0.49\textwidth]{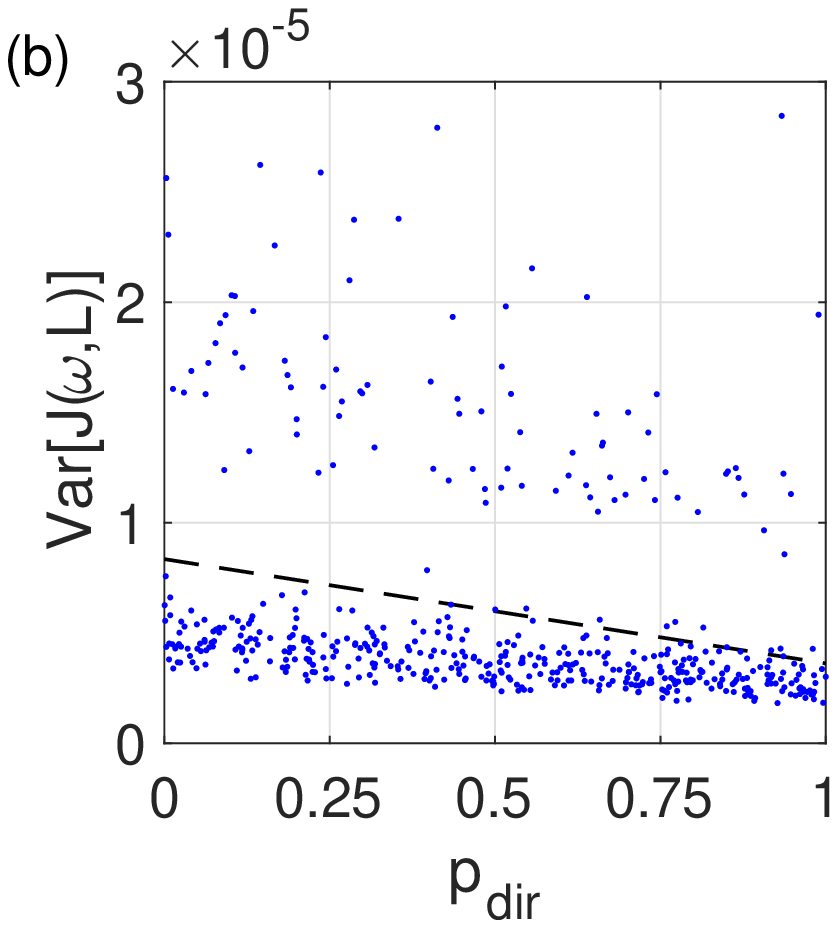}
\caption{{\bf Directedness: Expectation and variance of the SAF}. (a) The expected value and (b) variance of the SAF for a collection of $500$ unweighted ER networks of size $N=500$ with mean degree $\langle k \rangle=10$ with directedness $p_{\text{dir}}$ drawn randomly and uniformly from the interval $[0,1]$. The frequencies are drawn as IID normal-distributed random variables. We also plot the least squares lines, given by $\mathbb{E}[J(\bm{\omega},L)]\approx-4.70\times10^{-3} p_{\text{dir}}+2.11\times10^{-2}$ and $\mathrm{Var}[J(\bm{\omega},L)]\approx-4.73\times10^{-6} p_{\text{dir}}+8.35\times10^{-6}$, respectively. }\label{fig:10}
\end{figure}

We proceed by investigating the effect of directedness on synchronization by plotting in \cref{fig:09} the Kuramoto order parameter $r$ from a collection of unweighted ER networks of size $N=500$ with mean degree $\langle k\rangle =10$ that are iteratively rewired to become more directed. In particular, we vary directedness from $p_{\text{dir}}=0$ (totally undirected, red) to $1$ (completely directed, blue). Each data point represents the average of $r$ over a collection of $50$ networks. To emphasize the improved synchronization properties of more directed networks, we also plot in the inset a zoomed-in view of the synchronization profile in the strongly synchronized regime. To better understand this behavior, we consider a collection of $500$ ER networks (also with $N=500$ and $\langle k\rangle=10$), each assigned a directedness $p_{\text{dir}}$ randomly and uniformly drawn from the interval $[0,1]$. In \cref{fig:10} (a) and (b), we plot the expected value and variance of the SAF, respectively. For each case we also plot the least squares line, which for the data collected here are given by $\mathbb{E}[J(\bm{\omega},L)]\approx-4.70\times10^{-3} p_{\text{dir}}+2.11\times10^{-2}$ and $\mathrm{Var}[J(\bm{\omega},L)]\approx-4.73\times10^{-6} p_{\text{dir}}+8.35\times10^{-6}$, respectively. Interestingly, while the both slopes are negatively sloped, only the expected value $\mathbb{E}[J(\bm{\omega},L)]$ displays a sufficiently strong trend, even though significant fluctuations are present. Thus, increasing the directedness of a network appears to have a significant impact on the expected value of the SAF, but relatively little effect on the variance.

\subsection{Degree-Frequency Correlations Promote Synchronization}
\label{subsec:04:03}

The last property we investigate  is the presence of correlations between local networks structure and natural frequencies, i.e., structure-dynamics assortivity~\cite{Newman2003PRE}. Previous work has shown that positive correlations between oscillators' nodal degrees and natural frequencies can give rise to explosive synchronization and other novel dynamical behaviors~\cite{Gomez2011PRL,Skardal2013EPL,Skardal2014PRE}. Moreover, positive correlations between oscillators' nodal degree and the absolute value of natural frequencies appear in oscillator networks engineered for optimal synchronization~\cite{Brede2008PLA,Papadopoulos2017Chaos,Pinto2015PRE,Skardal2014PRL,Skardal2016Chaos}. However, especially in this latter case, the observation of correlations in optimized oscillator networks does not imply that such correlations generically promote synchronization properties. Here, we use the framework of \cref{sec:03} to investigate this specific relationship. 

From the viewpoint of uncertain frequencies developed herein, correlations between network structures and specific instances of natural frequencies may arise from properties related to either the means or the variances or both of these latent variables. We focus on two possible kinds of correlations: (i) correlations between the oscillators' nodal degrees $k_i$ and natural frequency magnitudes $|\mu_i|$ and (ii) correlations between the oscillators' nodal degrees $k_i$ and the natural frequency variances $\sigma_i^2$. For a given oscillator network structure such correlations may be measured using the Pearson correlation coefficient
\begin{align}
\rho_{x,y}=\frac{\sum_{n=1}^N(x_i-\langle x\rangle)(y_i-\langle y\rangle)}{\sqrt{\sum_{n=1}^N(x_i-\langle x\rangle)^2}\sqrt{\sum_{n=1}^N(y_i-\langle y\rangle)^2}}\label{eq:04:12}
\end{align}
for generic nodal properties $x_i$ and $y_i$ for $i=1,\dots,N$, so that degree-frequency mean and degree-frequency variance correlations may be measured using $\rho_{k,|\mu|}$ and $\rho_{k,\sigma^2}$, respectively. 

For simplicity we consider two different cases. To study degree-frequency mean correlations we let means $\mu_i$ be drawn randomly and uniformly from the interval $[-5,5]$ and natural frequencies be drawn independently (so that covariances are zero) with equal variance $\mathrm{Var}(\omega)$, in which case $\mathbb{E}[J(\bm{\omega,L})]$ and $\mathrm{Var}[J(\bm{\omega},L)]$ are described by \cref{cor:03:04} and \cref{eq:03:09}, respectively. To study degree-frequency variance correlations we let means $\mu_i=\mu=0$ be identically zero and natural frequencies be drawn independently (so that covariances are zero) with variances $\sigma_i^2$ drawn randomly and uniformly from the interval $[1,10]$, in which case $\mathbb{E}[J(\bm{\omega,L})]$ and $\mathrm{Var}[J(\bm{\omega},L)]$ are described by \cref{cor:03:03} and \cref{eq:03:08}, respectively. We  investigate the effect that degree-frequency mean and degree-frequency variance correlations have in collection of ER networks of size $N=500$ with mean degree $\langle 10\rangle$. For each network, after means $\mu_i$ or variances $\sigma_i^2$ are initially drawn, we rearrange them on the network to attain a correlation coefficient $\rho$ randomly drawn from the interval $[-1,1]$. 

\begin{figure}[htbp]
\centering
\includegraphics[width=0.49\textwidth]{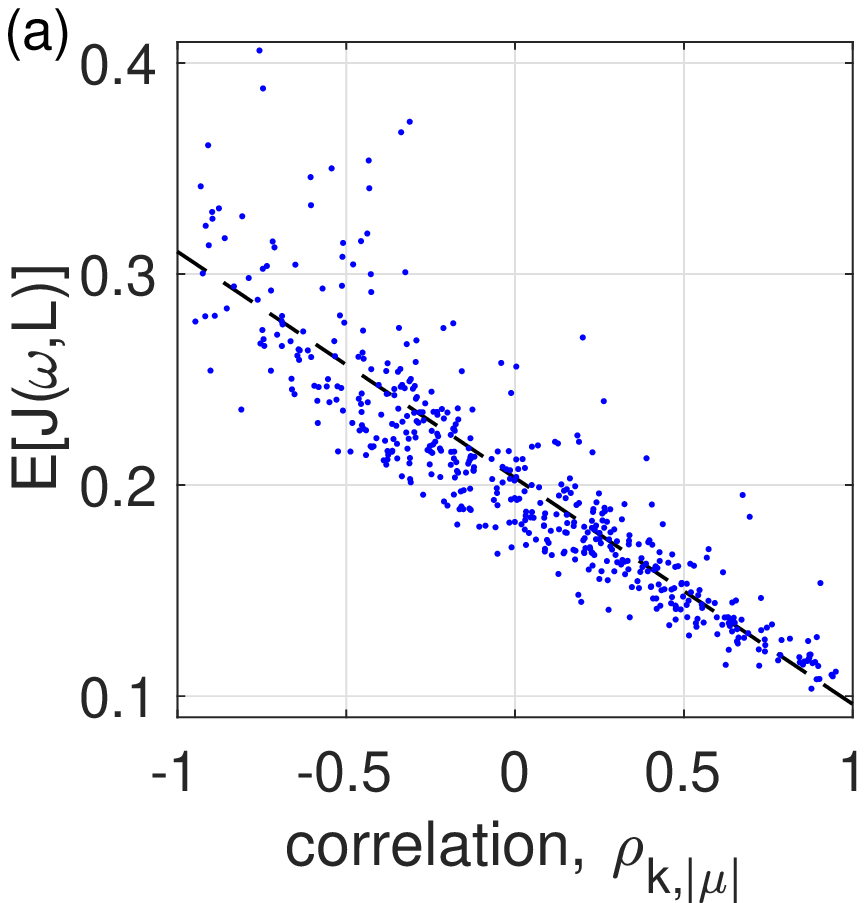}
\includegraphics[width=0.49\textwidth]{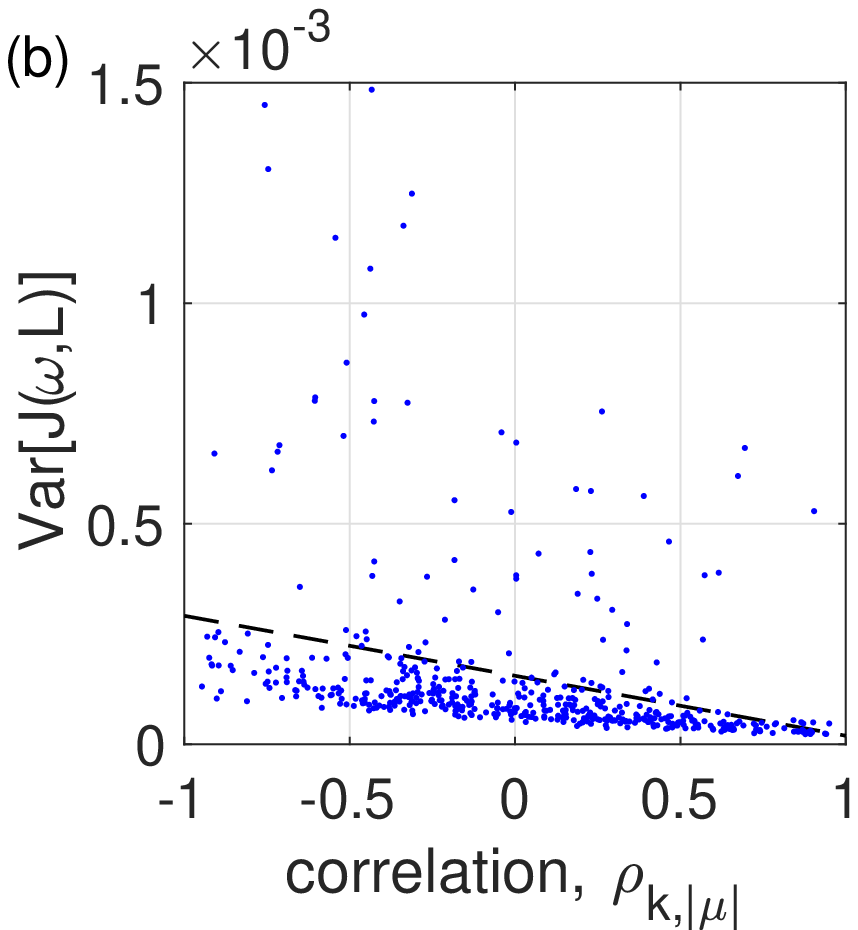}
\caption{{\bf Degree-frequency mean correlations: Expectation and variance of the SAF}. (a) The expected value and (b) variance of the SAF for a collection of $500$ unweighted, undirected ER networks of size $N=500$ with mean degree $\langle k \rangle=10$ with frequency means $\mu_i$ drawn randomly and uniformly from the interval $[-5,5]$ and allocated to randomly generate a degree-frequency mean correlation coefficient $\rho_{k,|\mu|}$ in the interval $[-1,1]$. We also plot the least squares lines, given by $\mathbb{E}[J(\bm{\omega},L)]\approx-1.07\times10^{-1}\rho_{k,|\mu|}+2.03\times10^{-1}$ and $\mathrm{Var}[J(\bm{\omega},L)]\approx-1.36\times10^{-4}\rho_{k,|\mu|}+1.55\times10^{-4}$, respectively.}\label{fig:11}
\end{figure}

In \cref{fig:11}, we present the result from degree-frequency mean correlations, plotting the expected value $\mathbb{E}[J(\bm{\omega},L)]$ and variance $\mathrm{Var}[J(\bm{\omega},L)]$ of the SAF versus the correlation coefficient $\rho_{k,|\mu|}$ in panels (a) and (b), respectively. We note a strong relationship between $\mathbb{E}[J(\bm{\omega},L)]$ and $\rho_{k,|\mu|}$, captured roughly by the line $\mathbb{E}[J(\bm{\omega},L)]\approx-1.07\times10^{-1}\rho_{k,|\mu|}+2.03\times10^{-1}$, indicating that positive correlations between degrees $k$ and natural frequency mean magnitudes $\mu$ significantly promotes synchronization. On the other hand, a weaker relationship exists between $\mathrm{Var}[J(\bm{\omega},L)]$ and $\rho_{k,|\mu|}$, captured roughly by the line $\mathrm{Var}[J(\bm{\omega},L)]\approx-1.36\times10^{-4}\rho_{k,|\mu|}+1.55\times10^{-4}$, indicating that positive correlations between degrees $k$ and natural frequency mean magnitudes $\mu$ decrease the variance of the SAF, but not significantly.

\begin{figure}[htbp]
\centering
\includegraphics[width=0.49\textwidth]{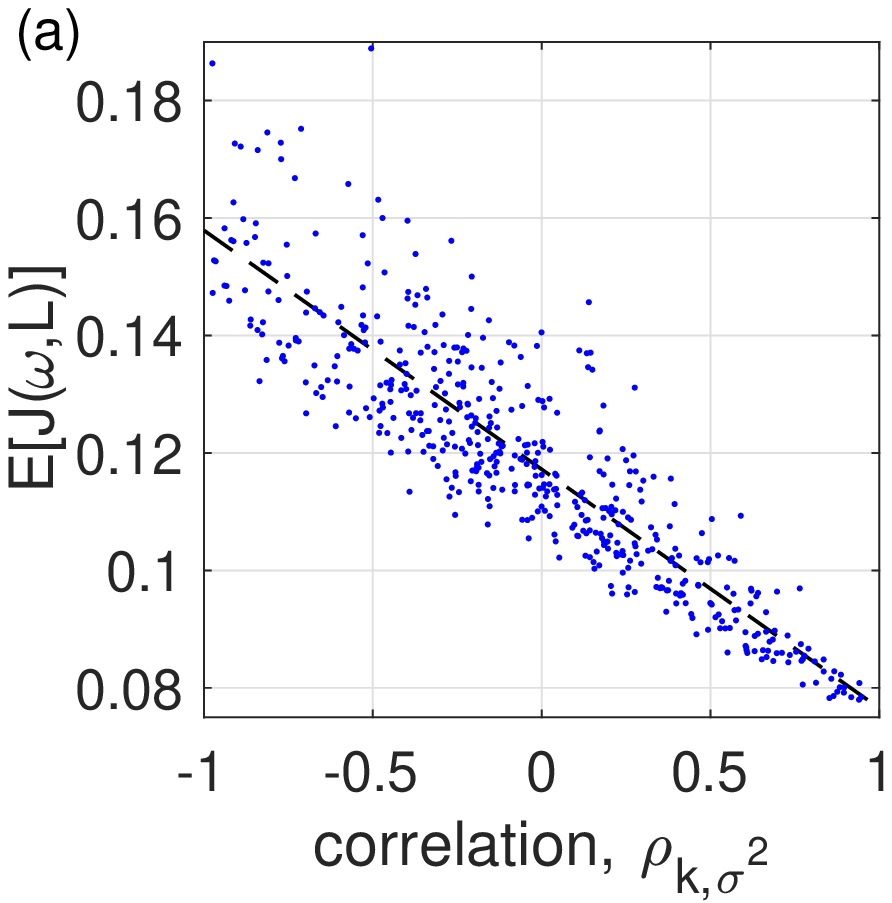}
\includegraphics[width=0.49\textwidth]{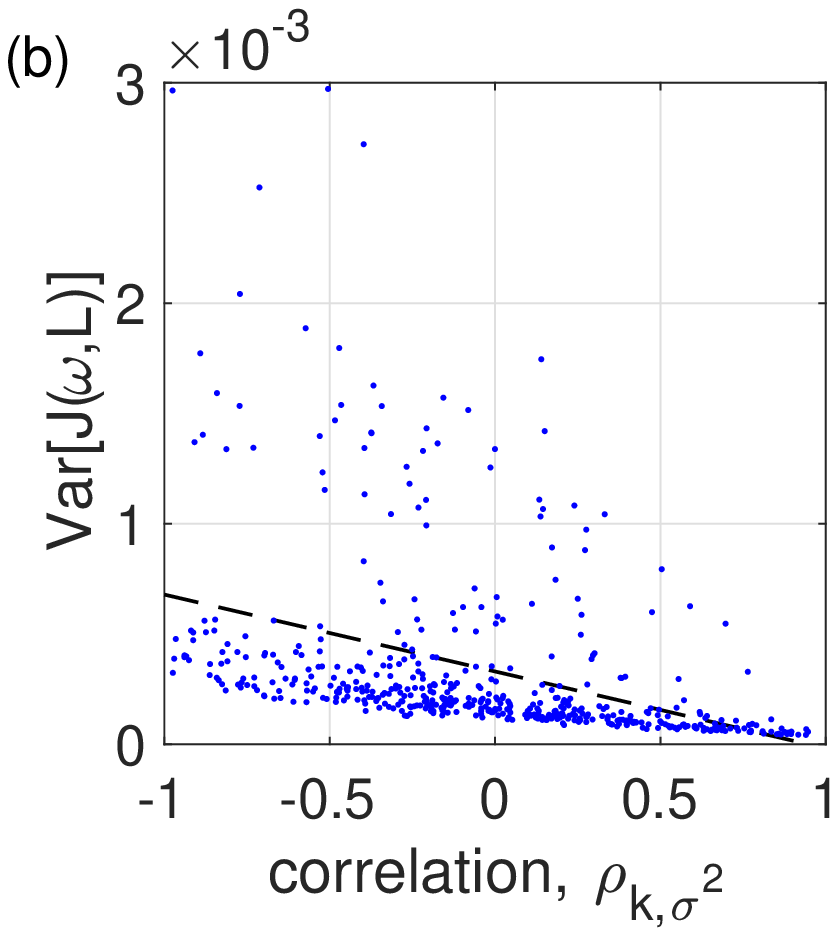}
\caption{{\bf Degree-frequency variance correlations: Expectation and variance of the SAF}. (a) The expected value and (b) variance of the SAF for a collection of $500$ unweighted, undirected ER networks of size $N=500$ with mean degree $\langle k \rangle=10$ with frequency variances $\sigma_i^2$ drawn randomly and uniformly from the interval $[1,10]$ and allocated to randomly generate a degree-frequency variance correlation coefficient $\rho_{k,\sigma^2}$ in the interval $[-1,1]$. We also plot the least squares lines, given by $\mathbb{E}[J(\bm{\omega},L)]\approx-4.07\times10^{-2}\rho_{k,\sigma^2}+1.17\times10^{-1}$ and $\mathrm{Var}[J(\bm{\omega},L)]\approx-3.49\times10^{-4}\rho_{k,\sigma^2}+3.30\times10^{-4}$, respectively.}\label{fig:12}
\end{figure}

Overall, we find a similar results for the case of degree-frequency variance correlations, presented in \cref{fig:11}, where we plot the expected value $\mathbb{E}[J(\bm{\omega},L)]$ and variance $\mathrm{Var}[J(\bm{\omega},L)]$ of the SAF versus the correlation coefficient $\rho_{k,\sigma^2}$ in panels (a) and (b), respectively. Again, we note a strong relationship between $\mathbb{E}[J(\bm{\omega},L)]$ and $\rho_{k,\sigma^2}$, captured roughly by the line $\mathbb{E}[J(\bm{\omega},L)]\approx-4.07\times10^{-2}\rho_{k,\sigma^2}+1.17\times10^{-1}$, indicating that positive correlations between degrees $k$ and natural frequency mean variances $\sigma^2$ significantly promotes synchronization. On the other hand, a weaker relationship exists between $\mathrm{Var}[J(\bm{\omega},L)]$ and $\rho_{k,\sigma^2}$, captured roughly by the line $\mathrm{Var}[J(\bm{\omega},L)]\approx-3.49\times10^{-4}\rho_{k,\sigma^2}+3.30\times10^{-4}$, indicating that, again, positive correlations between degrees $k$ and natural frequency mean variances $\sigma^2$ decrease the variance of the SAF, but not significantly.

\section{Discussion}\label{sec:05}

In this paper, we have investigated the behavior of the synchrony alignment function, or SAF, (see \cref{subsec:02:02} for a definition), which is a framework for evaluating how the interplay of structural network properties and internal dynamics (i.e., natural frequencies) shape the overall synchronization properties of heterogeneous oscillator networks. Extending previous research of the SAF \cite{Skardal2014PRL,Skardal2016Chaos,Skardal2015PRE2,Skardal2016Chaos,Taylor2016SIAP,Skardal2016Chaos} that treats the oscillators' natural frequencies as deterministic parameters, in \cref{sec:03} we considered the oscillators' natural frequencies as random variables and derived analytical expressions for the expectation and  variance of the SAF. Theorems~\ref{th:03:02} and \ref{th:03:07} and their corollaries characterize the expectation and variance of the SAF for when the frequencies are uncertain---they are drawn from a distribution with given mean and variance. 
Knowledge of the precise distribution is not required. 

An important insight from Theorem~\ref{th:03:02} is that the expected SAF yields two terms (one depends on the mean frequencies $\bm{\mu}$ and the other depends on the covariance matrix $\Sigma$), and one can consider  optimizing these two terms independently. The first term is the SAF, $J(\bm{\mu},L)$, with the expected frequencies as the first argument, rather than the particular frequencies $\bm{\omega}$. This provides theoretical support for why one can optimize the SAF using approximated frequencies, even when the exact frequencies are unknown [see also figure 6.2(c) in  \cite{Taylor2016SIAP} for a related numerical experiment]. One can optimize $J(\bm{\mu},L)$ under a variety of system constraints on $L$ and/or $\bm{\mu}$ using the same techniques that were explored in previous research on the SAF \cite{Skardal2014PRL,Skardal2016Chaos,Taylor2016SIAP}. For example, one can consider fixing the network structure and tuning the oscillators' frequencies, fixing the oscillators' frequencies and tuning the network structure design, or tuning both the network structure and the oscillators' frequencies. Optimizing the second term in Theorem~\ref{th:03:02}, $\frac{1}{N}\mathrm{Tr}[U(S^{\dagger})^2U^T \Sigma]$ is less straightforward and one's strategy may differ depending on the assumptions about $\Sigma$. 

Rather then explore the application of these results for optimization, however, we have applied them to gain fundamental insight into which network structure and dynamical properties promote synchronization. In particular, in \cref{sec:04} we provided theoretical and experimental for how link weight delocalization, link directedness, and degree-frequency correlations all enhance the synchronization properties of heterogeneous oscillator systems. For the case of IID-distributed random frequencies, we showed the expectation and variance of the SAF are proportional to the second and fourth moments of the spectral density (singular values in the case of directed networks and eigenvalues in the case of undirected networks). Intuitively, these spectral moments can be interpreted as measures for spectral heterogeneity, and our framework predicts synchronization is promoted by spectral concentration, i.e., a lack of spectral heterogeneity. This finding establishes an interesting connection to the study of synchronization for identical oscillators \cite{Barahona2002PRL,Nishikawa2010PNAS,Pecora1998PRL,Sun2009EPL}, for which the eigenratio --- one  measure for spectral heterogeneity --- is a widely accepted measure for the synchronizability of a system \cite{Barahona2002PRL}. Our findings for heterogeneous-but-uncertain frequencies are therefore in agreement with, but notably different  from this theory. We expect many applications can be benefited by treating the oscillator frequencies as uncertain --- that is rather than identical or deterministic --- and  the results presented herein provide an important step in this direction.

\appendix

\section{Proof of Theorem \cref{th:03:01}}\label{app:A}
\begin{proof}
This result is well known and is given as Theorem 5.2a in \cite{Rencher2008}. To provide the reader insight, we provide a simple proof here:
\begin{align}
E[ \bm{y}^T X\bm{y}] &= E\left[\sum_{i=1}^n \sum_{j=1}^n  X_{i,j} \bm{y}_i \bm{y}_j\right]\\
&= \sum_{i=1}^n \sum_{j=1}^n X_{i,j}E [\bm{y}_i \bm{y}_j]\\
&= \sum_{i=1}^n \sum_{j=1}^n X_{i,j}(\Sigma_{i,j}+\mu_i\mu_j)\\
 &= \sum_{i=1}^n \sum_{j=1}^n X_{i,j}\Sigma_{j,i} +\sum_{i=1}^n \sum_{j=1}^n X_{i,j}\mu_i\mu_j\\
&= \sum_{i=1}^n [ X\Sigma]_{i,i} + \bm{ \mu}^T X \bm{ \mu}\\
&= \textrm{Tr}[ X \Sigma] + \bm{ \mu}^T X\bm{ \mu}.
\end{align}
\end{proof}

\section{Proof of Theorem \cref{th:04:03}}\label{app:B}

We begin by considering, not the sequence of matrices $L^{(m)}$, but $\widehat{L}^{(m)}=L^{(m)}/w$. Note that the new matrices $\widehat{L}^{(m)}$ have off-diagonal entries that are either $0$ or $-1$, and therefore can be treated with \cref{th:04:02}. However, the new node degrees associated with the matrix $\widehat{L}^{(m)}$ is $k/\ell$, so the spectral densities limit to the distribution
\begin{align}
\widehat{\rho}_m(\widehat{\lambda}) \to \widehat{\rho}(\widehat{\lambda})=
\begin{cases}
\frac{\displaystyle k/\ell \sqrt{4 (k/\ell -1)  - (\widehat{\lambda}-k/\ell)^2}}{\displaystyle  2 \pi \left[(k^2/\ell^2 - (\widehat{\lambda}-k/\ell)^2)\right]} ,& \text{if $|\widehat{\lambda}-k/\ell| \leq 2
  \sqrt{k/\ell-1},$}\\ 
0 ,& \text{otherwise.}
\end{cases}\label{eq:04:04}
\end{align}
Now, since each $L^{(m)}=\ell\widehat{L}^{(m)}$ we have that $\widehat{\lambda}_n^{(m)}=\lambda_n^{(m)}/\ell$. Thus, we have that as $m\to\infty$ the spectral density functions $\rho_m(\lambda)\to\rho(\lambda)$, where 
\begin{align}
\rho(\lambda)&=\left|\frac{d\widehat{\lambda}}{d\lambda}\right|\widehat{\rho}(\lambda/\ell)\label{eq:04:05}\\
&=\begin{cases}
\left|\frac{1}{\ell}\right|\frac{\displaystyle k/\ell \sqrt{4 (k/\ell -1)  - (\lambda/\ell-k/\ell)^2}}{\displaystyle  2 \pi \left[k^2/\ell^2 - (\lambda-k)^2/\ell^2\right]} ,& \text{if $\left|\frac{\lambda}{\ell}-\frac{k}{\ell}\right| \leq 2\sqrt{\frac{k}{\ell}-1},$}\\ 
0 ,& \text{otherwise,}
\end{cases}\label{eq:04:06}\\
&=\begin{cases}
\frac{\displaystyle k \sqrt{4 (k/\ell -1)  - (\lambda/\ell-k/\ell)^2}}{\displaystyle  2 \pi \left[k^2 - (\lambda-k)^2\right]} ,& \text{if $|\lambda-k| \leq 2\sqrt{\ell k-\ell^2},$}\\ 
0 ,& \text{otherwise,}
\end{cases}\label{eq:04:07}\\
&=\begin{cases}
\frac{\displaystyle k \sqrt{4 (\ell k -\ell^2)  - (\lambda-k)^2}}{\displaystyle  2 \pi \ell\left[k^2 - (\lambda-k)^2\right]} ,& \text{if $|\lambda-k| \leq 2\sqrt{\ell k-\ell^2},$}\\ 
0 ,& \text{otherwise,}
\end{cases}\label{eq:04:08}
\end{align}
which completes the proof.


\bibliographystyle{siamplain}
\bibliography{references}
\end{document}